\documentclass[11pt]{article}
\usepackage{fullpage}

\usepackage{amsmath}
\usepackage{amsthm}
\usepackage{amssymb}
\usepackage{xspace}
\usepackage[ruled,vlined,linesnumbered]{algorithm2e}
\usepackage{ifpdf}
\ifpdf    
\usepackage{hyperref}
\else    
\usepackage[hypertex]{hyperref}
\fi

\newtheorem{theorem}{Theorem}[section]
\newtheorem{lemma}[theorem]{Lemma}
\newtheorem{corollary}[theorem]{Corollary}

\newtheorem{claim}[theorem]{Claim}
\newtheorem{definition}{Definition}[section]

\newtheorem{remark}[definition]{Remark}

\SetKwComment{tct}{}{}
\newcommand{\myparagraph}[1]{\paragraph{#1.}\xspace}

\newcommand{\given}{\mid}
\newcommand {\roundup}   [1] {{\lceil {#1} \rceil}}
\newcommand {\rounddown} [1] {{\lfloor {#1} \rfloor}}
\newcommand {\E}       {\mathbb{E}}

\def\cs{\ensuremath{\mathcal{C}}\xspace}
\def\p{\ensuremath{\mathcal{P}}\xspace}

\def\F{\ensuremath{\mathcal{F}}\xspace}
\def\sse{\subseteq}

\def\cal{\mathcal}

\def\MMKP{Min--Max $k$--Parti\-tio\-ning\xspace}
\def\MKP{Min--Sum $k$--Parti\-tio\-ning\xspace}
\def\MMMC{Min-Max-Multiway-Cut\xspace}
\def\mmc{{Min--Max Cut}\xspace}
\def\SSE{Small-Set Expansion\xspace}
\def\UC{{$\rho$--Unbalanced Cut}\xspace}
\def\WSSE{{Weighted Small-Set Expansion}\xspace}
\def\MBWC{{Weighted $\rho$-Unbalanced Cut}\xspace}
\def\BalancedCut{Balanced Cut\xspace}
\def\SparsestCut{Sparsest Cut\xspace}
\def\MinimumBisection{Minimum Bisection\xspace}
\def\MultiwayCut{Multiway Cut\xspace}

\def\opt{\ensuremath{\mathsf{OPT}}}
\def\veps{\varepsilon}
\def\whp{with high probability\xspace}
\def\ala{\`a la\xspace}
\def\eqdef{\mathrel{\mathop:}=}
\newcommand{\minn}[1]{\min\{#1\}}
\newcommand{\maxx}[1]{\max\{#1\}}
\newcommand{\ignore}[1]{}
\newcommand{\tuple}[1]{{\langle{#1}\rangle}}
\def\Racke{R\"{a}cke\xspace}
\def\compactify{\itemsep=0pt \topsep=0pt \partopsep=0pt \parsep=0pt}

\newcommand{\rnote}[1]{}
\newcommand{\vnote}[1]{}
\newcommand{\knote}[1]{}
\newcommand{\rsnote}[1]{}

\newcommand{\junk}[1]{}

\DeclareMathOperator{\diam}{diam}
\DeclareMathOperator{\supp}{supp}
\DeclareMathOperator{\argmax}{argmax}

\newcommand {\bbR}    {\mathbb{R}}

\newcommand {\calS}   {{\cal{S}}}
\newcommand {\calP}   {{\cal{P}}}

\newcommand{\initOneLiners}{%
    \setlength{\itemsep}{0pt}
    \setlength{\parsep }{0pt}
    \setlength{\topsep }{0pt}
}
\newenvironment{OneLiners}[1][\ensuremath{\bullet}]
    {\begin{list}
        {#1}
        {\initOneLiners}}
    {\end{list}}

\title{{\bf Min-Max Graph Partitioning and Small Set Expansion}}
\author{ Nikhil Bansal%
\thanks{Eindhoven University of Technology, The Netherlands. E-mail: \texttt{n.bansal@tue.nl}}
\and Uriel Feige\thanks{Weizmann Institute of Science, Rehovot, Israel. Work supported in part by The Israel Science
Foundation (grant \#873/08). Email:  \texttt{uriel.feige@weizmann.ac.il} }
\and Robert Krauthgamer%
\thanks{Weizmann Institute of Science, Rehovot, Israel.
Work supported in part by The Israel Science Foundation (grant \#452/08), and by a Minerva grant.
Email: \texttt{robert.krauthgamer@weizmann.ac.il}}
\and  Konstantin Makarychev%
\thanks{IBM T.J. Watson Research Center, P.O. Box 218, Yorktown Heights, NY 10598. Email: \texttt{konstantin@us.ibm.com}}
\and Viswanath Nagarajan%
\thanks{IBM T.J. Watson Research Center, P.O. Box 218, Yorktown Heights, NY 10598. Email: \texttt{viswanath@us.ibm.com}}
\and Joseph (Seffi) Naor%
\thanks{Computer Science Dept., Technion, Haifa, Israel. Email: \texttt{naor@cs.technion.ac.il}}
\and Roy Schwartz%
\thanks{Computer Science Dept., Technion, Haifa, Israel. Email: \texttt{schwartz@cs.technion.ac.il}}
}

\begin{document}
\maketitle

\begin{abstract}
We study graph partitioning problems from a \emph{min-max} perspective, in which an input graph on $n$ vertices should
be partitioned into $k$ parts, and the objective is to minimize the maximum number of edges leaving a single part. The
two main versions we consider are where
 the $k$ parts need to be of equal-size, and where they must separate a set of
$k$ given terminals. We consider a common generalization of these two problems, and design for it an $O(\sqrt{\log
n\log k})$-approximation algorithm. This improves over an $O(\log^2 n)$ approximation for the second version due to
Svitkina and Tardos~\cite{ST04}, and roughly $O(k\log n)$ approximation for the first version that follows from other
previous work. We also give an improved $O(1)$-approximation algorithm for graphs that exclude any fixed minor.

Our algorithm uses a new procedure for solving the \SSE problem. In this problem, we are given a graph $G$ and the goal
is to find a non-empty set $S\subseteq V$ of size $|S| \leq \rho n$ with minimum edge-expansion. We give an
$O(\sqrt{\log{n}\log{(1/\rho)}})$ bicriteria approximation algorithm for the general case of \SSE, and $O(1)$
approximation algorithm for graphs that exclude any fixed minor.
\end{abstract}

\section{Introduction} \label{sec:intro}

We study graph partitioning problems from a {\em min-max} perspective. Typically, graph partitioning problems ask for a
partitioning of the vertex set of an undirected graph under some problem-specific constraints on the different parts,
e.g., balanced partitioning or separating terminals, and the objective is {\em min-sum}, i.e., minimizing the total
weight of the edges connecting different parts. In the {\em min-max} variant of these problems, the goal is different
--- minimize the weight of the edges leaving a single part, taking the maximum over the different parts. A canonical
example, that we consider throughout the paper, is the \MMKP problem: given an undirected graph $G=(V,E)$ with
nonnegative edge-weights and $k\ge 2$, partition the vertices into $k$ (roughly) equal parts $S_1,\ldots,S_k$ so as to
minimize $\max_{i} \delta(S_i)$, where $\delta(S)$ denotes the sum of edge-weights in the cut $(S,V\setminus S)$. We
design a bicriteria approximation algorithm for this problem. Throughout, let $w:E\to \bbR^+$ denote the edge-weights
and let $n=|V|$.

Min-max partitions arise naturally in many settings. Consider the following application in the context of cloud
computing, which is a special case of the general graph-mapping problem considered in~\cite{BKNZ11} (and also implicit
in other previous works~\cite{VN-embed-sigcomm08,VN-embed-infocom06,VN-embed-infocom09}). There are $n$ processes
communicating with each other, and there are $k$ machines, each having a bandwidth capacity $C$. The goal is to
allocate the processes to machines in a way that balances the load (roughly $n/k$ processes per machine), and meets the
outgoing bandwidth requirement. Viewing the processes as vertices and the traffic between them as edge-weights, we get
the \MMKP problem. In general, balanced partitioning (either min-sum or min-max) is at the heart of many heuristics
that are used in a wide range of applications, including VLSI layout, circuit testing and simulation, parallel
scientific processing, and sparse linear systems.

Balanced partitioning, particularly in its min-sum version, has been studied extensively during the last two decades,
with impressive results and connections to several fields of mathematics, see e.g.
\cite{LR99,ENRS99,LLR95,AR98,ARV08,KNS09,LN06,CKN09}. The min-max variants, in contrast, have received much less
attention. Previously, no approximation algorithm for the \MMKP problem was given explicitly, and the approximation
that follows from known results is not smaller than
$O(k\sqrt{\log n})$.%
\footnote{One could reduce the problem to the min-sum version of $k$-partitioning. The latter admits bicriteria
approximation $O(\sqrt{\log n\log k})$ \cite{KNS09}, but the reduction loses another factor of $k/2$. Another
possibility is to repeatedly remove $n/k$ vertices from the graph, paying again a factor of $k/2$ on top of the
approximation in a single iteration, which is, say, $O(\log n)$ by \cite{Racke08}. } We improve this dependence on $k$
significantly.

\smallskip
An important tool in our result above is an approximation algorithm for the \SSE (SSE) problem. This problem was
suggested recently by Raghavendra and Steurer~\cite{RS10} (see also \cite{RST10,RSTdraft}) in the context of the unique
games conjecture. Recall that the \emph{edge-expansion} of a subset $S\sse V$ with $0<|S|\leq \tfrac12 |V|$ is
$$\Phi(S) \eqdef \frac{\delta(S)}{|S|}.$$
The input to the SSE problem is an edge-weighted graph and $\rho\in(0,\tfrac12]$, and the goal is to compute
$$  \Phi_\rho \eqdef \min_{|S|\leq \rho n} \Phi(S). $$
Raghavendra, Steurer and Tetali~\cite{RST10} designed for SSE an algorithm that approximates the expansion within
$O(\sqrt{(1/\Phi_\rho)\log (1/\rho)})$ factor of the optimum, while violating the bound on $|S|$ by no more than a
constant factor (namely, a bicriteria approximation). Notice that the approximation factor depends on $\Phi_\rho$; this
is not an issue if every small set expands well, but in general $\Phi_\rho$ can be as small as $1/\textrm{poly}(n)$, in
which case this guarantee is quite weak.

One can achieve a true approximation of $O(\log{n})$ for SSE using \cite{Racke08}, for any value of $\rho$.%
\footnote{For very small values of $\rho$, roughly $\rho n \le O(\log^2 n)$, a better approximation ratio is known
\cite{FKN03}.} If one desires a better approximation, then an approximation of $O(\sqrt{\log{n}})$ using \cite{ARV08}
can be achieved at the price of slightly violating the size constraint, namely a bicriteria approximation algorithm.
However, unlike the former which works for any value of $\rho$, the latter works only for $\rho = \Omega (1)$. In our
context of min-max problems we need the case $\rho = 1/k$, where $k=k(n)$ is part of the input. Therefore, it is
desirable to extend the $O(\sqrt{\log{n}})$ bound of \cite{ARV08} to a large range of values for $\rho$.

\subsection{Main Results}
Our two main results are bicriteria approximation algorithms for the \MMKP and SSE problems, presented below. The
notation $O_\veps(t)$ hides multiplicative factors depending on $\veps$, i.e., stands for $O(f(\veps)\cdot t)$.

\begin{theorem}
\label{thm:bal} For every positive constant $\varepsilon>0$, \MMKP admits a bicriteria approximation of $\big(
O_{\varepsilon}( \sqrt{\log{n}\log k}),\ 2+\varepsilon \big)$.
\end{theorem}

This theorem provides a polynomial-time algorithm that \whp\ outputs a partition $S_1,\ldots,S_k$ such that \ $\max_i
|S_i| \leq (2 + \varepsilon) \tfrac{n}{k}$ \ and \ $\max_i \delta(S_i) \leq O(\sqrt{\log{n}\log k}) \opt$, where $\opt$
is the optimal min-max value of partitioning into $k$ equal-size parts. (The guarantee on part size can be improved
slightly to $2-\frac1k+\varepsilon$). This result is most interesting in the regime $1\ll k\ll n$.

\begin{theorem}
\label{thm:sse} For every positive constant $\varepsilon >0$, \SSE admits a bicriteria approximation of $\big(
O_{\varepsilon} (\sqrt{\log{n}\log{(1/\rho)}}), 1+\varepsilon\big)$.
\end{theorem}

This theorem provides a polynomial-time algorithm that \whp\ outputs a set $S$ of size $0<|S|\leq (1+\varepsilon )\rho
n$ whose edge-expansion is $\delta (S)/|S| = O(\sqrt{\log{n}\log{(1/\rho)}}) \opt$, where $\opt$ is the minimum
edge-expansion over all sets of size at most $\rho n$. Our algorithm actually handles a more general version, called
\WSSE, which is required in Theorem \ref{thm:bal}. We defer the precise details to Section \ref{sec:wsse}.

\subsection{Additional Results and Extensions} \label{sec:MoreResults}

\myparagraph{\UC} Closely related to the SSE problem is the following \UC problem: The input is again a graph $G=(V,E)$
with nonnegative edge-weights and a parameter $\rho\in(0,\tfrac12]$, and the goal is to find a subset $S\sse V$ of size
$|S|=\rho n$ that minimizes $\delta (S)$. The relationship between this problem and SSE is similar to the one between
\BalancedCut and \SparsestCut, and thus Theorem \ref{thm:sse} yields the following result.

\begin{theorem}
\label{thm:uc} For every constant $0<\varepsilon<1$, the \UC problem admits a bicriteria approximation of
$\big(O_\veps(\sqrt{\log n \log (1/\rho)}), \Omega(1), 1+\veps\big)$.
\end{theorem}

This theorem says that there is a polynomial-time algorithm that \whp finds $S\subseteq V$ of size $\Omega(\rho n) \leq
|S| \leq (1+\varepsilon)\rho n$ and value $\delta (S) \leq O_{\varepsilon} (\sqrt{\log{n}\log{(1/\rho)}}) \opt$, where
$\opt$ is the value of an optimal solution to \UC. This result generalizes the bound of \cite{ARV08} from $\rho=\Omega
(1)$ to {\em any} value of $\rho\in(0,\tfrac12]$. Our factor is better than the $O(\log n)$ true approximation ratio
that follows from~\cite{Racke08}, at the price of slightly violating the size constraint. Our algorithm actually
handles a more general version, called \MBWC, which is required in Theorem \ref{thm:bal}. We defer the precise details
to Section~\ref{sec:SSE2UC}.

\myparagraph{\MMMC} We also consider the following \MMMC problem, suggested by Svitkina and Tardos~\cite{ST04}: the
input is an undirected graph with nonnegative edge-weights and $k$ terminal vertices $t_1,\ldots,t_k$, the goal is to
partition the vertices into $k$ parts $S_1,\ldots,S_k$ (not necessarily balanced), under the constraint that each part
contains exactly one terminal, so as to minimize $\max_i \delta(S_i)$. They designed an $O(\alpha \log
n)$--approximation algorithm for this problem, where $\alpha$ is the approximation factor known for \MinimumBisection.
Plugging $\alpha = O(\log n)$, due to R\"acke~\cite{Racke08}, the algorithm of Svitkina and Tardos achieves $O(\log^2
n)$-approximation. Using a similar algorithm to the one in Theorem~\ref{thm:bal}, we obtain a better approximation
factor.

\begin{theorem}
\label{thm:mcut} \MMMC admits an $O(\sqrt{\log{n}\log k})$--approximation algorithm.
\end{theorem}

Somewhat surprisingly, we show that removing the dependence on $n$ for \MMMC (even though no balance is required)
appears hard, which stands in contrast to its min-sum version, known as \MultiwayCut, which admits
$O(1)$--approximation \cite{CKR00,KKSMY04}. The idea is to show that it would imply a similar independence of $n$ for
the min-sum version of $k$-partitioning, thus for large but constant $k$, we would get an $(O(1), O(1))$-bicriteria
approximation for \MKP, which seems unlikely based on the current state of art~\cite{ARV08,AR06,KNS09}.

\begin{theorem}\label{thm:multiway-hard}
If there is a $k^{1-\varepsilon}$--approximation algorithm for \MMMC for some constant $\varepsilon>0$, then there is a
$(k^2,\gamma)$ bicriteria approximation algorithm for \MKP with $\gamma\le 3^{2/\varepsilon}$.
\end{theorem}

Additionally, we also consider a common generalization of \MMKP and \MMMC, which we call \mmc. In fact we obtain
Theorem~\ref{thm:mcut} as a special case of our result for \mmc.

\myparagraph{Excluded-minor graphs} Finally, we obtain an improved approximation -- {\em constant factor} -- for SSE in
graphs excluding a fixed minor.

\begin{theorem} \label{thm:sse-planar}
For every constant $\varepsilon>0$, \SSE admits:
\begin{OneLiners}
\item bicriteria approximation of $\big(O_\veps(r^2),\ 1+\veps\big)$ on graphs excluding a $K_{r,r}$-minor.
\item bicriteria approximation of $\big(O_\veps(\log g),\ 1+\veps\big)$ on graphs of genus $g\ge 1$.
\end{OneLiners}
\end{theorem}
These bounds extend to the \UC problem, and by plugging them  into the proof of Theorems~\ref{thm:bal} and
\ref{thm:mcut}, we achieve an improved approximation ratio of $O(r^2)$ for \MMKP and \MMMC in graphs excluding a
$K_{r,r}$-minor.

\subsection{Techniques}

For clarity, we restrict the discussion here mostly to our main application, \MMKP. Our approach has two main
ingredients. First, we reduce the problem to a weighted version of SSE, showing that an $\alpha$ (bicriteria)
approximation for the latter can be used to achieve $O(\alpha)$ (bicriteria) approximation for \MMKP. Second, we design
an $O_\veps(\sqrt{\log n \log (1/\rho)})$ (bicriteria) approximation for weighted SSE (recall that in our applications
$\rho = 1/k$).

Let us first examine SSE, and assume for simplicity of presentation that $\rho =1/k$. Note that SSE bears obvious
similarity to both \BalancedCut and min-sum $k$--partition (its solution contains a single cut with a size condition,
as in \BalancedCut, but the size of this cut is $n/k$ similarly to the $k$ pieces in min-sum $k$--partition). Thus, our
algorithm is inspired by, but different from, the approximation algorithms known for these two problems
\cite{ARV08,KNS09}. As in these two problems, we use a semidefinite programming (SDP) relaxation to compute an
$\ell_2^2$ metric on the graph vertices. However, new spreading constraints are needed since SSE is highly asymmetric
in its nature --- it contains only a {\em single} cut of size $n/k$. We devise a randomized rounding procedure based on
the {\em orthogonal separator} mechanics, first introduced by Chlamtac, Makarychev, and Makarychev \cite{CMM06} in the
context of unique games. These ideas lead to an algorithm that computes a cut $S$ of {\em expected} size $|S|\leq
O(n/k)$ and of {\em expected} cost $\delta (S)\leq O(\sqrt{\log n\log k})$ times the SDP value. An obvious concern is
that both properties occur in only expectation and might be badly correlated, e.g., the expected edge-expansion
$\E[\delta(S)/|S|]$ might be extremely large. Nevertheless, we prove that with good probability, $|S|=O(n/k)$ and $
\delta(S)/|S|$ is sufficiently small.

For SSE on excluded-minor and bounded-genus graphs, we give a better approximation guarantees, of a constant factor, by
extending the notion of orthogonal separators to linear programs (LPs) and designing such low-distortion ``LP
separators'' for these special graph families. The proof uses the probabilistic decompositions of Klein, Plotkin, and
Rao~\cite{KPR93} and Lee and Sidiropoulos~\cite{LS10}.  We believe that this result may be of independent interest. Let
us note that the LP formulation for SSE is not trivial and requires novel spreading constraints. We remark that even on
planar graphs, the decomposition of \Racke \cite{Racke08} suffers an $\Omega(\log n)$ loss in the approximation
guarantee, and thus does not yield $o(\log n)$ ratio for SSE on this class of graphs.

Several natural approaches for designing an approximation algorithm for Min--Max $k$--Parti\-tio\-ning fail. First,
reducing the problem to trees \ala \Racke \cite{Racke08} is not very effective, because there might not be a single
tree in the distribution that preserves {\em all} the $k$ cuts simultaneously. Standard arguments show that the loss
might be a factor of $O(k\log{n})$ in the case of $k$ different cuts.
Second, one can try and formulate a relaxation for the problem. However, the natural linear and semidefinite
relaxations both have large integrality gaps. As a case study, consider for a moment \MMMC. The standard linear
relaxation of Calinescu, Karloff and Rabani~\cite{CKR00} was shown by Svitkina and Tardos~\cite{ST04} to have an
integrality gap of $k/2$. In Appendix~\ref{sec:IntegralityGap} we extend this gap to the semidefinite relaxation that
includes all $\ell_2^2$ triangle inequality constraints.
A third attempt is to repeatedly remove from the graph, using SSE, pieces of size $\Theta (n/k)$. However, by removing
the ``wrong'' vertices from the graph, this process might arrive at a subgraph where every cut of $\Theta(n/k)$
vertices has edge-weight greater by a factor of $\Theta(k)$ than the original optimum (see
Appendix~\ref{app:greedy-bad} for details). Thus, a different approach is needed.

Our approach is to use multiplicative weight-updates on top of the algorithm for weighted SSE. This yields a collection
$\calS$ of sets $S$, all of size $|S|=\Theta(n/k)$ and cost $\delta(S)\leq O(\sqrt{\log n \log k}) \opt$, that covers
every vertex $v\in V$ at least $\Omega(n/k)$ times. (Alternatively, this collection $\calS$ can be viewed as a
fractional solution to a configuration LP of exponential size.) Next, we randomly sample sets $S_1,\dots, S_t$ from
$\calS$ till $V$ is covered, and derive a partition given by $P_1=S_1$, $P_2=S_2\setminus S_1$, and in general
$P_i=S_i\setminus (\cup_{j<i}S_j)$. This step is somewhat counter-intuitive, since the sets $P_i$ may have very large
cost $\delta(P_i)$ (because a set $P_i$ might be a strict subset of a set $S_{i'}$). We show that the total expected
boundary of the partition is not very large, i.e., $\E[\sum_i \delta(P_i)]\leq O(k \sqrt{\log n \log k}) \opt$.  Then,
we start fixing the partition by the following local operation: find a $P_i$ violating the constraint $\delta(P_i)
\leq O(\sqrt{\log n \log k}) \opt$, replace it with the unique $S_i$ containing it, and adjust other sets $P_j$
accordingly. Somewhat surprisingly, we prove that this local fixing procedure terminates (quickly). Finally, the
resulting partition consists of sets $P_i$, each of which satisfies the necessary properties, but now the number of
these sets might be very large. So the last step is to merge small sets together. We show that this can be done while
maintaining simultaneously the constraints on the sizes and on the costs of the sets.

\myparagraph{Organization} We first show in Section \ref{sec:wsse} how to approximate \WSSE (in both general and
excluded-minor graphs). We then show in Section \ref{sec:SSE2UC} that an approximation algorithm for \WSSE also yields
one for \MBWC. In Section \ref{sec:minmax-bal} we present an approximation algorithm for \MMKP that uses the
aforementioned algorithm for \UC (and in turn the one for \WSSE). The common generalization of both \MMKP and \MMMC,
\mmc, appears in Section~\ref{sec:FE}. Theorem \ref{thm:multiway-hard} is proved in Section~\ref{sec:Hardness}.


\section{Approximation Algorithms for Small Set Expansion}
\label{sec:wsse}

In this section we design approximation algorithms for the \SSE problem. Our main result is for general graphs and uses
an SDP relaxation. It actually holds for a slight generalization of the problem, where expansion is measured with
respect to vertex weights (see Definition \ref{defn:WSSE} and Theorem \ref{thm:WSSE}). We further obtain improved
approximation for certain graph families such as planar graphs (see Section \ref{sec:SSE_LP}).

To simplify notation, we shall assume that vertex weights are normalized: we consider measures $\mu$ and $\eta$ with
$\mu(V)=\eta(V)=1$. We denote $\mu(u)=\mu(\{u\})$ and $\eta(u)=\eta(\{u\})$. We let $(V,w)$ denote a complete
(undirected) graph on vertex set $V$ with edge-weight $w(u,v)=w(v,u)\ge 0$ for every $u\neq v\in V$. In our context,
such $(V,w)$ can easily model a specific edge set $E$, by simply setting $w(u,v)=0$ for every non-edge $(u,v)\notin E$.
Recall that we let $\delta(S) \eqdef \sum_{u\in S,v\in V\setminus S} w(u,v)$ be the total weight of edges crossing the
cut $(S,V\setminus S)$, and further let $w(E)$ denote the total weight of all edges.

\begin{definition}[\WSSE]
\label{defn:WSSE} Let $G=(V,w)$ be a graph with nonnegative edge-weights, and let $\mu$ and $\eta$ be two
 measures on the vertex set $V$ with $\mu(V)=\eta(V)=1$.
The \emph{weighted small set expansion} with respect to
$\rho \in (0,1/2]$ is
$$\Phi_{\rho,\mu,\eta} (G) \quad \eqdef \quad
\min \left\{ \frac{\delta(S)}{w(E)}\times \frac{1}{\eta(S)}
  \ :\ \eta(S)>0,\ \mu(S)\leq \rho \right\}.
$$
\end{definition}

\begin{theorem}[Approximating SSE] \label{thm:WSSE}
(I) For every fixed $\varepsilon>0$, there is a polynomial-time algorithm that
given as input an edge-weighted graph $G=(V,w)$,
two measures $\mu$ and $\eta$ on $V$ ($\mu(V)=\eta(V)=1$), and some $\rho\in(0,1/2]$,
finds a set $S\subset V$ satisfying
$\eta(S)>0$, $\mu(S)\leq (1+\varepsilon) \rho$
and
\begin{equation}\label{eq:sse_distort}
\frac{\delta(S)}{w(E)}\times \frac{1}{\eta(S)} \quad \leq \quad D\times \Phi_{\rho,\mu,\eta} (G),
\end{equation}
where $D= O_{\varepsilon} (\sqrt{\log n\log (1/\rho)})$.

\noindent (II) When the input contains in addition a parameter $H\in(0,1)$, the algorithm finds a non-empty set
$S\subset V$ satisfying $\mu(S)\leq (1+\varepsilon) \rho$, $\eta(S) \in [\Omega(H), 2(1+ \veps) H]$, and
\begin{equation}
\frac{\delta(S)}{w(E)}\times \frac{1}{\eta(S)} \quad
  \leq \quad D\times \min \left\{ \frac{\delta(S)}{w(E)}\times \frac{1}{\eta(S)}
     \ :\ \eta(S)\in[H,2H] ,\ \mu(S)\leq \rho  \right\},
\end{equation}
where $D= O_{\varepsilon} (\sqrt{\log n\log (\maxx{1/\rho, 1/H})})$.
\end{theorem}

We prove part I of the theorem in Section \ref{sec:alg1}, and part II in Section~\ref{sec:alg2}. These algorithms
require the following notion of $m$-orthogonal separators due to Chlamtac, Makarychev, and Makarychev~\cite{CMM06}.

\begin{definition}[Orthogonal Separators]
Let $X$ be an $\ell_2^2$ space (i.e., a collection of vectors satisfying $\ell_2^2$ triangle inequalities).
We say that a distribution over subsets of $X$
is
an $m$-orthogonal separator of $X$
with distortion $D$, probability scale $\alpha>0$ and separation threshold  $\beta < 1$ if the following
conditions hold for $S\subset X$ chosen according to this distribution:
\begin{OneLiners}
\item For all $u\in X$ we have $\Pr (u\in S) = \alpha \, \| u\|^2$.
\item For all $u,v\in X$ with  $\|u - v\|^2\geq \beta
\min (\| u\|^2, \|v\|^2)$,
$$\Pr (u \in S \text{ and } v \in S) \quad \leq \quad \frac{\minn{\Pr(u \in S),  \Pr(v \in S)}}{m}.$$
\item For all $u,v\in X$ we have
$\Pr(I_S(u) \neq  I_S(v))\leq \alpha  D \times \|u - v\|^2$, where $I_S$ is the indicator function of the set $S$.
\end{OneLiners}
\end{definition}

\begin{theorem}[\cite{CMM06}]
There exists a polynomial-time randomized algorithm that given a set of vectors $X$, positive number $m$, and $\beta
<1$ generates $m$-orthogonal separator with distortion $D=O_{\beta}(\sqrt{\log |X| \log m})$ and scale $\alpha \geq
1/p(|X|)$ for some polynomial $p$.
\end{theorem}

In the original paper~\cite{CMM06}, the second requirement
in the definition of orthogonal separators was slightly different,
however, exactly the same algorithm and proof works in our case:
If
$\|u-v\|^2\geq \beta \|u\|^2$ and $\|u\|^2 \leq \|v\|^2$, then $\langle u, v\rangle =
(\|u\|^2 + \|v\|^2 - \|u-v\|^2)/2 \leq ((1-\beta)\|u\|^2 +\|v\|^2)/2 \leq (1-\beta/2)\|v\|^2$.
Then, by Lemma~4.1 in~\cite{CMM06}, $\langle \varphi (u), \varphi(v)\rangle \leq (1-\beta/2)$;
hence $\|\varphi(u) - \varphi(v)\|^2 \geq \beta > 0$ and, in Corollary~4.6,
$\|\psi(u) - \psi(v)\| \geq 2\gamma = \sqrt{\beta}/4 > 0$.

\subsection{Algorithm I: \SSE in General Graphs}
\label{sec:alg1}

We now prove part I of Theorem \ref{thm:WSSE}.

\myparagraph{SDP Relaxation}  In our relaxation we introduce a vector $\bar v$ for every vertex $v\in V$. In the intended solution of the SDP corresponding to the optimal solution $S\subset V$, $\bar{v} = 1$ (or, a fixed unit vector $e$), if $v\in S$; and
$\bar{v} = 0$, otherwise. The objective is to minimize the fraction of cut edges
$$\min \;\;\;\frac{1}{w(E)} \;\sum_{(u,v)\in E}   w(u,v)\, \|\bar{u}-\bar{v}\|^2.$$
We could constrain all vectors $\bar v$ to have length at most 1,
i.e. $\|\bar{v}\|^2 \leq 1$,
but it turns out our algorithm never uses this constraint.
We require that the vectors $\{\bar v:\ v\in V\}\cup\{0\}$
satisfy $\ell_2^2$ triangle inequalities i.e., for every
$u,v,w\in V$, 
$\|\bar u- \bar w\|^2 \leq \|\bar u- \bar v\|^2 +\|\bar v - \bar w\|^2$,
$\|\bar u\|^2 \leq \|\bar u - \bar v\|^2 +\|\bar v\|^2$,
$\|\bar u - \bar w\|^2 \leq \|\bar u\|^2 +\|\bar w\|^2$.
Suppose now that we have approximately guessed the measure $H$ of the optimal solution
$H\leq \eta(S)\leq 2H$ (this step is not necessary but it simplifies the exposition; in fact, we could
simply let $H=1$, since the SDP is otherwise homogeneous). This can be done
since the measure of every set $S$ lies in the range from $\eta(u)$ to $n\eta(u)$, where $u$ is
the heaviest element in $S$, hence $H$ can be chosen from the set $\{2^t \eta(u): u\in V, t=0,\cdots,\rounddown{\log_2 n} \}$ of size $O(n \log n)$.  Then we add a constraint
\begin{equation}\label{eq:measureH}
\sum_{v\in V} \|\bar{v}\|^2 \eta(v) \quad \geq \quad H.
\end{equation}
We denote  $\eta(u) = \eta(\{u\})$ and $\mu(u) = \mu(\{u\})$.
Finally, we introduce new spreading constraints: for every $u\in V$,
$$\sum_{v \in V} \mu(v)\cdot \minn{\| \bar{u} - \bar{v}\|^2,\|u\|^2} \quad \geq \quad (1-\rho) \|\bar u\|^2.$$
(Alternatively, we could use a slightly simpler, almost  equivalent  constraint $\sum_{v \in V} \langle \bar{u} , \bar{v} \rangle \mu(v)  \leq \rho \|\bar u\|^2$. We chose to use the former formulation
because an analogous constraint can be written in a linear program, see Section~\ref{sec:SSE_LP}.)
In the intended solution this constraint is satisfied, since if $u\in S$,
then $\bar{u}=1$ and the sum above equals $\mu(V\setminus S) \geq 1-\rho$.
If $u\notin S$, then $\bar{u}=0$ and both sides of the constraint equal $0$.

The SDP relaxation used in our algorithm is presented below in its entirety. Note that the second constraint can be
written as $\tuple{u,v}\leq \|u\|^2$, and the third constraint can be written as $\langle u, v \rangle \geq 0 $.
\medskip
\begin{equation} \label{SDP:SSE}
\framebox{
\begin{minipage}[b]{0.8\linewidth}
\medskip\par
\begin{equation*}
\mathbf{min} \qquad\qquad \frac{1}{w(E)} \;\sum_{(u,v)\in E} w(u,v)\; \|\bar u - \bar v\|^2
\end{equation*}
\vskip -15pt
\begin{eqnarray*}
\mbox{{\bf s.t.} }\,\, \|\bar u - \bar w\|^2 +  \|\bar w - \bar v\|^2 &\geq& \|\bar u - \bar v\|^2, \qquad \,\, \forall u,v,w \in V,\\
  \|\bar u- \bar v\|^2&\geq& \|\bar u\|^2 - \|\bar v\|^2,\quad \forall u,v\in V,\\
  \|\bar u\|^2 + \|\bar v\|^2 &\geq& \|\bar u- \bar v\|^2,\qquad \,\, \forall u,v\in V,\\
  \sum_{v \in V} \mu(v)\cdot \minn{\| \bar{u} - \bar{v}\|^2,\|\bar u\|^2} &\geq& (1-\rho) \|\bar u\|^2, \quad\, \forall \,\, u \in V,\\
  \sum_{v\in V} \|\bar{v}\|^2 \eta(v) \quad &\geq& \quad H.
\end{eqnarray*}
\end{minipage}\medskip\par}\end{equation}

We now describe the approximation algorithm.

\myparagraph{Approximation Algorithm}
We first informally describe the main idea behind the algorithm. The algorithm solves the SDP relaxation and obtains a set of vectors $\{\bar{u}\}_{u\in V}$. Now it samples an orthogonal separator, a random set $S\subset V$, and returns it. Assume for the moment that $\alpha = 1$. Since $\Pr(v\in S)=\|\bar{v}\|^2$, we get $\E[ \eta (S)] \geq H$. The expected size of the cut is at most $D\times SDP$ by the third property of orthogonal separators; and thus the sparsity is at most
$D\times SDP/H\leq 2D\times OPT$. The second property
of orthogonal separators guarantees that if $\bar{u}\in S$, then the vectors
that are far from $\bar{u}$,
a very small fraction will belong to $S$ (since the conditional probability
$\Pr(\bar{v}\in S\given \bar{u}\in S)\leq 1/m$ is very small). And by the spreading constraints,
at most $(1+\varepsilon)\rho$ fraction of vectors (w.r.t. the measure $\mu$)
is close to $\bar{u}$. Hence, the total expected measure of $S$ is at most
$(1+\varepsilon)\rho + 1/m\leq (1+2\varepsilon)\rho$. We now proceed to the formal argument.

We may assume that $\varepsilon$ is sufficiently small i.e.,
 $\varepsilon \in (0,1/4)$. The approximation algorithm guesses approximate value of
the weight $H$: $H \leq \eta(S) \leq 2H$. Set the length of all vectors $\bar{u}$ with $\eta(u)> 2H$ to be 0. Solves
the SDP and obtains a set of vectors $X=\{\bar{v}\}_{v\in V}$. Then, it finds an orthogonal separator $S$ with $m =
\max(\varepsilon^{-1}\rho^{-1})$ and $\beta = \varepsilon$. For convenience, we let $S$ be the set of vertices
corresponding to vectors belonging to the orthogonal separator rather than the vectors themselves. The algorithm
repeats the previous step $\roundup{\alpha^{-1} n^2}$ times (recall $\alpha$ is the probabilistic scale of the
orthogonal separator) and outputs the best $S$ satisfying $0< \mu(S) < (1+10\varepsilon)\rho$. With an exponentially
small probability no $S$ satisfies this constraint, in which case, the algorithm outputs an arbitrary set satisfying
constraints.

\myparagraph{Analysis} We first estimate the probability of the event ``$u\in S$ and $\mu(S) < (1+10\varepsilon)
\rho$'' for a fixed vertex $u\in V$. Let $A_u=\{v: \|\bar{u} - \bar{v} \|^2 \geq \beta \|\bar u\|^2\}$ and $B_u=\{v: \|
\bar{u} - \bar{v} \|^2  < \beta \|\bar u\|^2\}$.
 We show that only a small
fraction of $A_u$ belongs to $S$, and that the set $B_u$ is small.

From the spreading constraint
$\sum_{v \in V} \min(\| \bar{u} - \bar{v}\|^2,\|u\|^2)\mu(v)\geq (1-\rho) \|\bar u\|^2$,
by the Markov inequality, we get that $\mu(B_u) \leq \rho/(1-\beta) \leq (1+2\varepsilon) \rho$. 
For an arbitrary $v\in A_u$
(for which $\bar v \neq 0$) write
$\| \bar u- \bar v\|^2 \geq \beta \|\bar u\|^2 \geq  \beta \min (\|\bar u\|^2,\|\bar v\|^2)$. 
By the second property of orthogonal separators, $\Pr(v\in S \given u \in S) \leq 1/m$, thus the expected measure
$\mu(A_u\cap S)$ is at most $\E \mu(A_u\cap S) \leq \varepsilon \rho$. Now, by the Markov inequality, given that $u\in
S$, the probability of the bad event ``$\mu(S) \geq (1+10\varepsilon)\rho$ (and, thus $\mu(A_u\cap S)\geq 8\varepsilon
\rho$)'' is at most $1/8$. Each vertex $u\in V$ belongs to $S$ with probability $\alpha \|\bar u\|^2$. Hence, $u\in S$,
and $\mu(S) < (1+10\varepsilon) \rho$ with probability at least $3/4 \; \alpha \|\bar u\|^2$.

Finally, we use the third property of orthogonal separators to bound the size of the cut $\delta(S)$
$$\E\delta(S)=
\sum_{(u,v)\in E} |I_S(u) - I_S(v)| w(u,v)
 \leq \alpha D \times
\sum_{(u,v)\in E} \|\bar u - \bar v\|^2  w(u,v) = \alpha D \times SDP \times w(E).$$
Here, as usual, $SDP$ denotes the value of the SDP solution; and
$D=O_{\varepsilon}(\sqrt{\log n \log (1/\delta)})$ is the distortion of $m$-orthogonal separators.

Define function
$$f(S) = \eta(S) - \frac{\delta (S)}{w(E)}\times \frac{H}{4D\times SDP},$$
if $|S| \neq \varnothing$, $\mu(S) < (1+10\varepsilon)\rho$, and $f(S)=0$, otherwise. The expectation
\begin{eqnarray*}
\E f(S) &\geq&
\sum_{u\in V} \frac{3\alpha \|\bar u\|^2 \eta(u)}{4} -
\frac{\alpha H}{4} \geq \frac{\alpha H}{2}.
\end{eqnarray*}

The random variable $f(S)$ is always bounded by $2nH$, thus with probability at least $\alpha/n$,
$f(S) > 0$. Therefore, with probability exponentially close to 1,
after $\alpha^{-1}  n^2$ iterations, the algorithm will find $S$ with $f(S)>0$.
Since $f(S)>0$, we get $\eta(S)>0$, $\mu(S) < (1+10\varepsilon)\rho$, and
$$\frac{\delta (S)}{w(E)}\times \frac{1}{\eta(S)} \leq 4D \times \frac{SDP}{H}.$$
This finishes the proof of part I since $SDP/(2H) \leq \Phi_{\rho,\mu,\eta} (G)$.

\subsection{Algorithm II: \SSE in General Graphs}
\label{sec:alg2}

We now prove part II of Theorem \ref{thm:WSSE}.
This algorithm uses an SDP relaxation similar to part I,
although we need a few additional constraints.
We write a constraint ensuring that ``$\eta(S) \leq 2H$'' (recall $H$ is an approximate
value of $\eta(S)$ in the optimal solution): we add spreading constraints for all $u\in V$,
$$\sum_{v \in V} \minn{\|\bar{u} - \bar{v} \|^2, \|\bar u\|^2} \;\eta(v) \leq 2H \|\bar u\|^2,$$
and we let $m=\maxx{\varepsilon^{-1}\rho^{-1}, H^{-1}\rho^{-1}}$.
We also require
\begin{equation}\label{eq:muleqrho}
\sum_{v \in V} \|v\|^2 \mu(v) \leq \rho.
\end{equation}

Algorithm II gets $H$, the approximate value of the measure $\eta(S)$,
as input, and thus does not need to guess it.

\begin{remark}\label{rem:terminals} To handle terminals in the extended version of the problem (see Section~\ref{sec:FE})
we guess which terminal $u\in T$ belongs to the optimal solution $S$ (if any), and set $\|\bar{u}\|=1$ and
$\|\bar{v}\|=0$ for $v\in T\setminus \{u\}$. Since an orthogonal separator never contains the zero vector, we will
never choose more than one terminal in the set~$S$.
\end{remark}

\textbf{Approximation Algorithm.}
The algorithm consists of many iterations of a slightly modified Algorithm~I. At every step
the algorithm obtains a set $S$ of vertices (returned by Algorithm I) and adds it to the set $T$, which is initially empty. Then, the algorithm removes vectors corresponding to $S$ from the set $X$, the SDP solution, and repeats the same procedure till
$\mu(T)\geq \rho/4$ or $\eta(T)\geq H/4$. In the end, the algorithm returns the set $T$
if $\mu(T)\leq \rho$ and $\eta(T)\leq H$, and the last set $S$ otherwise.

The algorithm changes the SDP solution (by removing some vectors), however we can ignore
these changes, since the objective value of the SDP may only decrease and all constraints
but~(\ref{eq:measureH}) are still satisfied. Since the total weight $\eta(T)$ of removed vertices
is at most $H/4$, a slightly weaker variant of constraint~(\ref{eq:measureH}) is satisfied. Namely,
\begin{equation*}
\sum_{u\in V} \|\bar{u}\|^2 \eta(u) \geq 3H/4.
\eqno(\ref{eq:measureH}')
\end{equation*}

We now describe the changes in Algorithm I: instead of $f$, we define function $f'$:
$$f'(S) = \eta(S) - \frac{\delta (S)}{w(E)}\times \frac{H}{4D\times SDP}
-
\frac{\mu(S)}{4\rho}\times H,$$
if $|S| \neq \varnothing$, $\mu(S) < (1+10\varepsilon)\rho$ and
$\eta (S)\leq (1+10\varepsilon) H$ and $f'(S)=0$, otherwise. Notice, that $f'$ has an
extra term comparing to $f$ and, in order for $f'(S)$ to be positive,
the constraint $\eta (S)\leq 2(1+10\varepsilon) H$ should be satisfied. The new
variant of Algorithm I, returns $S$, once $f'(S)> 0$.

The same argument as before shows that for any given $u\in V$ conditional on $u\in S$,
$\mu (S)\leq (1+10\varepsilon) \rho$ and $\eta (S)\leq 2(1+10\varepsilon) H$ with probability
at least $3/4$. Then, using a new constraint~(\ref{eq:muleqrho}), we get $\E\mu(S) \leq \alpha \rho$.
Hence, the expectation
$$
\E f'(S) \geq \frac{3\alpha \times 3/4\;H}{4} - \frac{\alpha H}{4} - \frac{\alpha H}{4} \geq \frac{\alpha H}{16}.
$$
Again, after at most $O(\alpha^{-1}n^2)$ iterations the algorithm will find $S$ with
$f'(S)> 0$ (and only with exponentially small probability fail)\footnote{In fact, now $f'(S)\leq 2H$, thus we need only $O(\alpha^{-1}n)$ iterations.}.
Then, $f'(S) > 0$ implies
\begin{equation}\label{eq:ineqSEta}
\frac{\delta (S)}{w(E)} \leq 4D \times \frac{SDP}{H} \eta(S);
\end{equation}
and
$\eta(S)\geq H \times \mu(S)/(4\rho)$.

The last inequality implies that at every moment $\eta(T)\geq H \times \mu(T)/(4\rho)$. Hence, if $\mu(T) \geq \rho/4$
(recall, this is one of the two conditions, when the algorithm stops), then $\eta(T) \geq H/16$. Therefore, if the
algorithm returns set $T$, then $\eta(T) \geq H/16$. If the algorithm returns set $S$ then either $\mu(S)\geq 3/4\,
\rho$ and thus $\eta(S) \geq 3H/16$ or $\eta(S)\geq 3/4\,H$.

Both, $\mu(T)$ and $\eta(T)$ are bounded from above by $\rho$ and $H$ respectively;
$\mu(S)$ and $\eta(S)$ are bounded from above by $(1+10\varepsilon)\rho$ and
$2(1+10\varepsilon)H$ respectively.

The inequality~(\ref{eq:ineqSEta}) holds for every set $S$ added in $T$, hence
this inequality holds for $T$.

\subsection{\SSE in Minor-Closed Graph Families}
\label{sec:SSE_LP}

In this subsection we prove Theorem \ref{thm:sse-planar}.
We start by writing an LP relaxation. For every vertex $u\in V$ we introduce a variable $x(u)$ taking values in $[0,1]$; and for
every pair of vertices $u,v\in V$ we introduce a variable $z(u,v)=z(v,u)$ also taking values in $[0,1]$. In the intended
integral solution corresponding to a set $S\subset V$, $x(u) = 1$ if $u\in S$, and $x(u) = 0$ otherwise; $z(u,v) =
|x(u) - x(v)|$.
(One way of
thinking of $x(u)$ is as the distance to some imaginary vertex $O$ that never belongs to $S$. In the SDP relaxation vertex $O$ is the origin.)
It is instructive to think of $x(u)$ as an analog of $\|\bar u\|^2$ and of
$z(u,v)$ as an analog of $\|\bar u- \bar v\|^2$.

It is easy to verify that LP \eqref{LP:SSE} below is a relaxation of the \SSE problem. It has a constraint saying that
$z(u,v)$ is a metric (or, strictly speaking, semi-metric). A novelty of the LP is in the third constraint, which is a
new spreading constraints for ensuring the size of $S$ is small.

\begin{equation} \label{LP:SSE}
\framebox{
\begin{minipage}[b]{0.80\linewidth}
\medskip
\begin{equation*}
\mathbf{min} \qquad\qquad \frac{1}{w(E)} \;\sum_{(u,v)\in E} w(u,v)\; \|\bar u - \bar v\|^2
\end{equation*}
\vskip -15pt
\begin{eqnarray*}
\mathbf{s.t.} \quad z(u,v)+z(v,w) &\geq& z(u,w), \qquad \quad \forall u,v,w \in V,\\
|x(u) - x(v)| & \leq& z(u,v), \qquad \quad \forall u,v \in V,\\
\sum_{v \in V} \mu(v)   \cdot \min{x(u), \,z(u,v)} &\geq & (1-\rho) x(u), \quad \forall u \in V,\\
x(u), z(u,v) &\in&  [0,1], \qquad \qquad \forall u,v \in V.\\
\end{eqnarray*}
\end{minipage}
}
\end{equation}

We introduce an analog of $m$-orthogonal separators for linear programming,
which we call LP separators.

\begin{definition}[LP separator]
\label{defn:l1sep}
Let $G=(V,E)$ be a graph, and let $\{x(u),z(u,v)\}_{u,v\in V}$ be a set of numbers. We say that a distribution over subsets of $V$
is
an LP separator of $V$
with distortion $D\ge1$, probability scale $\alpha>0$ and separation threshold  $\beta\in(0,1)$ if the following
conditions hold for $S\subset V$ chosen according to this distribution:
\begin{OneLiners}
\item For all $u\in V$, $\Pr (u\in S) = \alpha \, x(u)$. 
\item For all $u,v\in V$ with  $z(u,v) \geq \beta \minn{x(u),x(v)}$,
$\Pr (u \in S \text{ and } v \in S) = 0$. 
\item For all $(u,v)\in E$,
$\Pr(I_S(u) \neq  I_S(v))\leq \alpha  D \times z(u,v)$, 
where $I_S$ is an indicator for the set $S$.
\end{OneLiners}
\end{definition}

Below we present an efficient algorithm for an LP separator: given a graph $G=(V,E)$ excluding $K_{r,r}$ as a minor, a
parameter $\beta\in(0,1)$, and a set of numbers $\{x(u),z(u,v)\}_{u,v\in V}$ satisfying the triangle inequalities
described above (but not necessarily the spreading constraints), the algorithm computes an LP separator with distortion
$O(r^2)$ (for genus $g$ graphs the distortion is $O(\log g)$).
This proves Theorem \ref{thm:sse-planar} as follows:
by replacing in the algorithms above
the SDP relaxation \eqref{SDP:SSE} with the LP relaxation \eqref{LP:SSE},
and the orthogonal separators with LP separators,
we obtain $O(r^2)$ approximation algorithm approximation algorithm for SSE
in $K_{r,r}$ excluded-minor graphs.
Combined with the framework in Section \ref{sec:minmax-bal}, we
consequently obtain an
$O(r^2)$-approximation algorithm for \MMKP and \MMMC on such graphs.

\myparagraph{Computing LP Separators} We now describe an algorithm that samples an LP separator (see
Definition~\ref{defn:l1sep}) with respect to a feasible solution to LP \eqref{LP:SSE}. We recall a standard notion of
low-diameter decomposition of a metric space, see e.g. \cite{Bartal96,GKL03,KR11} and references therein.

Let $(V,d)$ be a finite metric space.
Given a partition $P$ of $V$ and a point $v\in V$,
we refer to the elements of $P$ as clusters,
and let $P(v)$ denote the cluster $S\in P$ that contains $v$,
so $v\in S\in \calP$.
A \emph{stochastic decomposition} of this metric
is a probability distribution $\nu$ over partitions $P$ of $V$.

\begin{definition} [Separating Decomposition]
Let $D,\Delta>0$.
A stochastic decomposition $\nu$ of a finite metric space $(V,d)$
is called a \emph{$D$-separating $\Delta$-bounded decomposition}
if it satisfies:
\begin{itemize} \compactify
\item For every partition $P\in\supp(\nu)$ and every cluster $S\in P$,
$$  \diam(S) \eqdef \max_{u,v\in S} d(u,v) \leq \Delta.$$
\item For every $u,v\in V$, the probability that a partition $P$
sampled from $\nu$ separates them is
$$\Pr_{P\sim\nu} [P(u) \neq P(v)] \leq D\cdot\frac{d(u,v)}{\Delta}.$$
\end{itemize}
\end{definition}

\begin{theorem}[\cite{KPR93,Rao99,FT03}]
Let $G=(V,E)$ be a graph excluding $K_{r,r}$ as a minor,
equipped with nonnegative edge-lengths.
Then the graph's shortest-path metric $d_G$ admits, for every $\Delta > 0$,
an $O(r^2)$-separating $\Delta$-bounded decomposition.
Moreover, there is a polynomial-time algorithm that samples
a partition from this distribution.
\end{theorem}
Lee and Sidiropoulos~\cite{LS10} show similarly for graphs
with genus $g\ge 1$ an $O(\log g)$-separating decomposition.
Alternative algorithms for both cases are shown in \cite{KR11}.

\begin{definition} [Probabilistic Partitioning]
\label{def:probpartdist}
Consider a graph $G=(V,E)$ and nonnegative numbers $\{x(u),z(u,v)\}_{u,v\in V}$.
We say that a distribution $\nu$ over partitions of $V$
is called a \emph{probabilistic partitioning} with distortion $D>0$
and separation threshold $\beta>0$ if the following
properties hold:
\begin{itemize} \compactify
\item For every edge $(u,v)\in E$ with $x(u)>0$:
$$\Pr_{P\sim\nu} (P(u) \neq P(v)) \leq D\cdot z(u,v)/x(u).$$
\item For every $u,v\in V$ with $z(u,v)\geq \beta x(u)$,
we have $P(u) \neq P(v)$ for all $P\in\supp(\nu)$.
\end{itemize}
\end{definition}

\begin{theorem}[Separating decomposition implies probabilistic partitioning]
Let $G=(V,E)$ be a graph that excludes $K_{r,r}$ as a minor,
and let $\{x(u),z(u,v)\}_{u,v\in V}$ satisfy
the first two constraints of LP \eqref{LP:SSE}.
Then for every $\beta\in (0,1]$, there is a probabilistic partitioning $\nu$
with distortion $D=O(r^{2}\beta^{-1})$ and separation threshold $\beta$.
\end{theorem}
\begin{proof}
We define new lengths $y(u,v) = \minn {1, z(u,v)/x(u), z(u,v)/x(v)}$ for all $u,v\in V$;
by convention, if $x(u)=0$ or $x(v)=0$ then define $y(u,v)=0$.
(We remark that a similar approach was used in~\cite{CKR00}).
These lengths
may violate the triangle inequality. Let $d:V\times V \to \bbR$ be the shortest path metric in
graph $G$ with edge lengths equal $y(u,v)$ (note: $y(u,v)$ is defined for all pairs $u,v\in V$; however,
to obtain $d$ we look only at $(u,v)\in E$). Clearly, for every edge $(u,v)\in E$, $d(u,v)\leq y(u,v)$. On the other hand, as we show below, for every $u,v\in V$, $d(u,v)\geq \tfrac23\, y(u,v)$.

\begin{claim}\label{cl:shpathmet}
For all $u,v\in V$ we have $d(u,v)\geq \tfrac23\, y(u,v)$.
\end{claim}
\begin{proof}
Pick two vertices $u,v\in V$ and consider an arbitrary path $u=w_1,w_2,\dots, w_N=v$. We prove that the length of the
path (in which the length of each edge $(w_i, w_{i+1})$ is $y(w_{i},w_{i+1})$) is at least $\tfrac23\, y(u,v)$. If the length
of the path is greater than $2/3$ we are done. Thus, we may assume that the lengths of all edges are at most $2/3<1$.
We also assume that $x(u)\geq x(v)$ and thus $y(u,v)\le z(u,v)/x(u)$. We have,
\begin{eqnarray*}
\sum_{i=1}^{N-1} y(w_i,w_{i+1}) &=&
\sum_{i=1}^{N-1} \frac{z(w_i,w_{i+1})}{\max(x(w_i),x(w_{i+1}))}\geq
\frac{1}{\max_{i}(x(w_i))} \sum_{i=1}^{N-1} z(w_i,w_{i+1})
\\ &\geq& \frac{z(u,v)}{\max_{i}(x(w_i))} \ge \frac{x(u)\;y(u,v)}{\max_{i}(x(w_i))}.
\end{eqnarray*}
The second inequality holds since $z(\cdot,\cdot)$ is a metric. If $x(u)/\max_{i}(x(w_i))\geq 2/3$,
we are done. Assume that for $j=\argmax x(w_j)$, $x(u)/x(w_j) < 2/3$. Then,
\begin{eqnarray*}
\sum_{i=1}^{N-1} y(w_i,w_{i+1})&\geq& \sum_{i=1}^{N-1} \frac{z(w_i,w_{i+1})}{x(w_j)} \geq
\sum_{i=1}^{N-1} \frac{|x(w_i) - x(w_{i+1})|}{x(w_j)} \\
&\geq& \frac{|x(w_{j}) - x(u)| + |x(w_{j}) - x(v)|}{x(w_j)} \geq 2\Bigl(1-\frac{x(u)}{x(w_j)}\Bigr)\geq \frac{2}{3}.
\end{eqnarray*}
This finishes the proof of Claim~\ref{cl:shpathmet}.
\end{proof}
We now apply the theorem of Klein, Plotkin, and Rao~\cite{KPR93} to the metric $d(u,v)$ and obtain a probabilistic
partition $\calP$ with $\Delta = \beta/3$ and $D'=O(r^{2})$. This partition satisfies the following properties.
\begin{itemize}
\item If $z(u,v) \geq \beta x(u)$ for $u,v\in V$, then either $x(u)\geq x(v)$ and hence
$y(u,v)\geq \beta$, or $x(v)\geq x(u)$, then (using $z(u,v) \geq x(v) - x(u)$)
$$z(u,v) \geq \frac{\beta}{2} (x(v) - x(u)) + \bigl(1-\frac{\beta}{2}\bigr) z(u,v)
\geq \frac{\beta}{2} (x(v) - x(u)) + \bigl(1-\frac{\beta}{2}\bigr) \beta x(u) \geq  \frac{\beta}{2} x(v).$$
Thus, $y(u,v)\geq \beta/2$ in either case and $d(u,v) \geq 2/3\,y(u,v) \geq \beta/3\equiv \Delta$
(by Claim~\ref{cl:shpathmet}).
\item For every $(u,v)\in E$,
$$\Pr (P_u\neq P_v) \leq D'\, \frac{d(u,v)}{\Delta} \leq \frac{D'}{\Delta}\,\frac{z(u,v)}{x(u)}.$$
\end{itemize}
The distortion $D$ equals $D'/\Delta = O(r^2/\beta)$.
\end{proof}

Given a solution for LP \eqref{LP:SSE} (the relaxation for SSE problem), we could proceed as follows: Construct a
probabilistic partition $\calP$ with distortion $D=O(r^2\beta^{-1})$ and some constant separation threshold $\beta\in
(0,1)$, then pick a random vertex $w\in V$ with probability $x(w)\eta(w)\left/\sum_{u}x(u)\eta(u)\right.$ and, finally,
output the cluster $P_w$. However, to highlight the similarity between this LP-algorithm and the previous SDP-algorithm
(for general graphs),  we give an algorithm for constructing LP separators, which in turn is used by the \SSE
algorithm.

\begin{theorem} \label{thm:LPsep2}
There exists an algorithm that given a graph $G=(V,E)$ with an excluded minor $K_{r,r}$, a set of numbers
$\{x(u),z(u,v)\}_{u,v\in V}$ satisfying the triangle inequality constraints, and a parameter $\beta\in [0,1]$, returns
an LP separator $S\subset V$ with distortion $D=r^2\beta^{-1}$ and separation threshold $\beta$.
\end{theorem}

{\textbf{Algorithm}}.
The algorithm samples a random partition $P$ with distortion $D=O(r^2\beta^{-1})$ and a
separation threshold $\beta$. For every $C\in P$, let
$$x_{\infty}(C) = \max_{u\in C}x(u).$$
The algorithm picks a random set $S\in P$ with probability $\Pr (S = C)=
x_{\infty}(C)/n$; and with the remaining probability
$$1-\frac{1}{n}\sum_{C\in P}x_{\infty}(C)\geq 1-\frac{1}{n}\sum_{u \in V}x(u)\geq 0,$$
the algorithm sets $S=\varnothing$.

Now, to guarantee that every vertex $u$ is chosen with probability
exactly $\alpha x(u)$, where $\alpha = 1/n$, the algorithm removes some
elements from $S$: it picks at random $t\in [0,1]$ and outputs set
$$S'= \{u\in S: x(u) \geq t x_{\infty}(S)\}.$$

{\textbf{Analysis}}.  Verify that $S'$ satisfies the properties of LP separators (with
$\alpha = 1/n$). For every $u\in V$,
$$\Pr(u\in S') =
\E_{P}\Big[\Pr(S = P(u)\given P)\cdot \Pr (x(u)\geq t x_{\infty}(P(u)))\given P\Big]
=\E_{P}\left[\frac{x_{\infty}(P(u))}{n}\cdot \frac{x(u)}{x_{\infty}(P(u))}\right]= \frac{x(u)}{n}.$$
Then, if $z(u,v)\geq \min (x(u),x(v))$, then $P(u)\neq P(v)$ and hence
$$\Pr(u,v\in S) \eqdef \Pr(P(u) = P(v) = S) = 0.$$
Finally,
$$
\Pr (u\in S', v\notin S') \leq \Pr (v\notin S\given u\in S')\Pr(u\in S') +
\Pr (x(v) \leq t x_{\infty}(S)\given u\in S')\Pr(u\in S').$$
We estimate the first term (using that $\nu$ has distortion $D=O(r^2\beta^{-1})$; see Definition~\ref{def:probpartdist})
$$
\Pr (u\in S, v\notin S) = \frac{x(u)}{n} \Pr(v\notin S \given u \in S) \leq \frac{x(u)}{n} \times D\frac{z(u,v)}{x(u)}
= \frac{D}{n} z(u,v),
$$
and then the second term
\begin{eqnarray*}
\Pr (x(v) \leq t x_{\infty}(S)\given u\in S')\Pr(u\in S') &=&
\frac{x(u)}{n} \Pr (x(v) \leq t x_{\infty}(P_u)\given x(u) \geq t x_{\infty}(P_u), S=P_u)\\
&=&\frac{x(u)}{n} \Pr (t \geq  x(v)/x_{\infty}(P_u)\given t \leq  x(u)/x_{\infty}(P_u))\\
&=&\frac{x(u)}{n}\times\frac{\min(0, x(u)-x(v))}{x(u)}\leq \frac{z(u,v)}{n}.
\end{eqnarray*}

This completes the proof of Theorem \ref{thm:LPsep2}.

\subsection{\texorpdfstring{From SSE to \UC}{From SSE to p-Unbalanced Cut}} \label{sec:SSE2UC}

\UC and SSE are equivalent, up to some constants, with respect to
bicriteria approximation guarantees.
Indeed, the two problems are related in the same way that
\BalancedCut and \SparsestCut are.
We refer the reader to \cite{LR99,RST10}, and omit details from this version of the paper.

Our intended application of approximating \MMKP (in Section \ref{sec:minmax-bal}),
requires a weighted version of the \UC problem, as follows.

\begin{definition}[\MBWC] \label{def:SmallHeavyCut}
The input to this problem is a tuple $\tuple{G,y,w,\tau,\rho}$,
where $G=(V,E)$ is a graph with
vertex-weights $y:V\rightarrow \mathbb{R}_+$, edge-costs $w:E\rightarrow \bbR_{\ge 0}$, and parameters $\tau,
\rho\in(0,1]$. The goal is to find $S\subseteq V$ of minimum cost $\delta(S)$ satisfying:
\begin{enumerate} \compactify
\item $y(S)\geq \tau\cdot y(V)$; and
\item $|S|\leq \rho\cdot n$.
\end{enumerate}
\end{definition}
The unweighted version of the problem (defined in Section \ref{sec:MoreResults}) has $\tau=\rho$ and unit
vertex-weights, i.e. $y(v)=1$ for all $v\in V$. We focus on the direction of reducing \MBWC to \WSSE, which is needed
for our intended application. Formally, we have the following corollary of Theorem \ref{thm:WSSE}. We use
$\opt_{\tuple{G,y,w,\tau,\rho}}$ to denote the optimal value of the corresponding weighted \UC instance.

\begin{corollary}[Approximating $\rho$-Unbalanced Cut]
\label{cor:sdp-subroutine}
For every $\varepsilon>0$, there exists a polynomial-time algorithm that given an instance $\tuple{G,y,w,\tau,\rho}$ of \MBWC, finds a set $S$ satisfying $|S|\leq \beta\rho n$, $y(S) \geq \tau/\gamma$ and
$\delta(S) \leq \alpha \cdot \opt_{\tuple{G,y,w,\tau,\rho}}$
for $\alpha = O_{\varepsilon} (\sqrt{\log n\log (\max(1/\rho, 1/\tau))})$,
$\beta=1+\varepsilon$ and
$\gamma=O(1)$.
\end{corollary}
\begin{proof}
Let $S^*$ be an optimal solution to $\tuple{G,y,w,\tau,\rho}$, note that $|S^*|\leq \rho n$,
$y(S^*)\geq \tau\cdot y(V)$ and $\delta(S^*) = \opt_{\tuple{G,y,w,\tau,\rho}}$ the optimal value of this instance. Define two measures on $V$
as follows. For any $S\sse V$, set $\mu(S) : = |S|/n$ and $\eta(S) \eqdef y(S)/y(V)$.

The algorithm guesses $H\geq \tau$ such that $H \leq \eta(S^*) \leq 2H$ (see Algorithm I above for an argument why we
can guess $H$). Then it invokes the algorithm from part II on $G$ with measures $\mu$ and $\eta$ as defined above, and
parameters $\rho,H$. The obtained solution $S$ satisfies $|S|=\mu(S)\cdot n\le (1+\varepsilon)\rho\,n$ and
$y(S)=\eta(S)\cdot y(V)\ge \Omega_{\varepsilon}(1)\, H \cdot y(V)\ge \Omega_{\varepsilon}(1)\, \tau \cdot y(V)$, since
$H\ge \tau$. Furthermore, $ \delta(S) \leq \alpha \cdot \delta(S^*)\cdot \eta(S)/\eta(S^*) \leq   \alpha \cdot
\delta(S^*)\cdot \Theta_\varepsilon(1)$, where $\alpha=O_{\varepsilon} (\sqrt{\log n\log (\max(1/\rho, 1/\tau))})$.
\end{proof}

\section{Min-max Balanced Partitioning}\label{sec:minmax-bal}

In this section, we present our algorithm for \MMKP, assuming a subroutine that approximates \MBWC (which is
essentially a rephrasing of \WSSE). Our algorithm for \MMKP follows by a straightforward composition of
Theorem~\ref{thm:covering} and Theorem~\ref{thm:aggreg} below. Plugging in for $(\alpha,\beta,\gamma)$ the values
obtained in Section \ref{sec:wsse} would complete the proof of Theorem~\ref{thm:bal}.

\subsection{Uniform Coverings}
\label{sec:mult}
We first consider a covering relaxation of \MMKP and solve it using multiplicative updates.
This covering relaxation can alternatively be viewed as a fractional solution
to a configuration LP of exponential size, as discussed further below.

Let $\cs = \left\{ S\sse V : \, |S|\le n/k \right\}$ denote all the vertex-sets that are feasible for a single
part. Note that a feasible solution in \MMKP corresponds to a partition of $V$ into $k$ parts, where each part belongs to
\cs. Algorithm \ref{alg:covMMC}, described below, {\em uniformly covers} $V$ using sets in \cs (actually a slightly larger family than \cs).
It is important to note that its output $\calS$ is a multiset.

\LinesNotNumbered\DontPrintSemicolon
\begin{algorithm}[!ht]\caption{Covering Procedure for \MMKP:} \label{alg:covMMC}
Set $t = 1$, and $y^1(v) = 1$ for all $v\in V$\;
 \While{$\sum_{v\in V} y^{t}(v) > 1/n$}{
   \tcp{Solve the following using algorithm from Corollary~\ref{cor:sdp-subroutine}.}
Let $S^t\sse V$ be the solution for \MBWC instance $\tuple{G,y^t,w,\frac{1}k,\frac{1}k}$.\; Set $\calS = \calS\cup
\{S^t\}$.\; \tcp{Update the weights of the covered vertices.} \For{\bf{every} $v\in V$}{Set $y^{t+1}(v) = \frac12\cdot
y^{t}(v)$ if $v\in S^t$, and $y^{t+1}(v) = y^{t}(v)$  otherwise.\;
 }
Set $t = t+1$.
 }
\Return $\calS$\;
\end{algorithm}

\begin{theorem}\label{thm:covering}
Running Algorithm \ref{alg:covMMC} on an instance of \MMKP
outputs $\calS$ that satisfies
(here $\opt$ denotes the optimal value of the instance):
\begin{enumerate} \compactify
\item For all $S\in \calS$ we have\;
$\delta(S)\le \alpha\cdot \opt$ and $|S|\le \beta\cdot n/k$.
\item For all $v\in V$ we have\;
$|\{S\in \calS : S\ni v\}|/|\calS|\ge 1/(5\gamma k)$.
\end{enumerate}
\end{theorem}
\begin{proof}
For an iteration $t$, let us denote $Y^{t} \eqdef \sum_{v\in V} y^{t}(v)$. The first assertion of the theorem is
immediate from the following claim.

\begin{claim}
Every iteration $t$ of Algorithm \ref{alg:covMMC} satisfies
$\delta(S^t)\le \alpha\cdot \opt$ and $|S^t|\le \beta\cdot n/k$.
\end{claim}
\begin{proof}
It suffices to show that the optimal value of the \MBWC instance $\langle
G, \, y^t,\, w,\, \frac{1}k,\, \frac{1}k\rangle$ is at most \opt. To see this, consider the optimal solution
$\{S^*_i\}_{i=1}^k$ of the original \MMKP instance. We have $|S^*_i|\le n/k$ and $w(\delta(S^*_i))\le \opt$ for
all $i\in [k]$. Since $\{S^*_i\}_{i=1}^k$ partitions $V$, there is some $j\in[k]$ with $y^t(S^*_j)\ge Y^t/k$. It now
follows that $S^*_j$ is a feasible solution to the \MBWC instance $\langle G, \, y^t,\, w,\, \frac{1}k,\, \frac{1}k
\rangle$, with objective value at most \opt,
which proves the claim.
\end{proof}

We proceed to prove the second assertion of Theorem~\ref{thm:covering}. Let $\ell$ denote the number of iterations of
the while loop, for the given \MMKP instance. For any $v\in V$, let $N_v$ denote the number of iterations $t$ with
$S^t\ni v$. Then, by the $y$-updates we have $y^{\ell+1}(v)=1/2^{N_v}$. Moreover, the termination condition implies
that $y^{\ell+1}(v)\le 1/n$ (since $Y^{\ell+1}\le 1/n$). Thus we obtain $N_v\ge \log_2 n$ for all $v\in V$. From the
approximation guarantee of the \MBWC algorithm, it follows that $y^t(S^t)\ge \frac1{\gamma\, k}\cdot Y^t$ in every
iteration $t$. Thus $Y^{t+1}=Y^t-\frac12 \cdot y^t(S^t)\le \left( 1-\frac1{2\gamma\, k} \right)\cdot Y^t$. This implies
that $Y^{\ell}\le \left( 1-\frac1{2\gamma\, k} \right)^{\ell-1} \cdot Y^1= \left( 1-\frac1{2\gamma\, k}
\right)^{\ell-1} \cdot n$. However $Y^\ell>1/n$ since the algorithm performs $\ell$ iterations. Thus, $\ell\le
1+4\gamma\,k\cdot \ln n\le 5\gamma\, k\cdot \log_2n$. This proves $|\{S\in \calS : S\ni v\}|/|\calS| = N_v/\ell\ge
(5\gamma)^{-1}k^{-1}$.
\end{proof}

\myparagraph{Alternative view: A configuration LP}

We now describe an alternate approach to finding a cover $\calS$. Given a bound $\lambda$ on the cost of any single
cut, define the set of feasible cuts as follows:
$$ {\cal{F}}_{\lambda} = \left\{ S\subseteq V ~: ~|S|\leq \frac{n}k, \,\, \delta (S)\leq \lambda \right\} .$$
We define a configuration LP for \MMKP as follows. There is a variable $x_S$ for each $S\in\F_{\lambda}$ indicating
whether/not cut $S$ is chosen.
\begin{equation} \label{LP:config}
\framebox{
$
\begin{array}{llll}
  {\cal P}(\lambda) =
  & \min
  & \displaystyle \sum _{S\in {\cal{F}}_{\lambda}} x_S
  \\
  & \mathrm{s.t.}
  & \displaystyle \sum _{S\in {\cal{F}}_{\lambda} : v\in S} x_S \geq 1
  & \forall v\in V
  \\
  & & x_S \geq 0
  & \forall S\in {\cal{F}}_{\lambda}
\end{array}
$}
\end{equation}
The goal is determine the smallest $\lambda>0$ such that ${\cal{P}}(\lambda) \leq k$. One can approximately solve this
using the dual formulation:
\begin{align} \label{LP:configDual}
\framebox{
$
\begin{array}{llll}
  {\cal{D}}(\lambda) =
  & \max
  & \displaystyle \sum_{v\in V} y_v
  \\
  & \mathrm{s.t.}
  & \displaystyle \sum _{v\in S} y_v \leq 1
  & \forall S\in {\cal{F}}_{\lambda}
  \\
  & & y_v\geq 0
  & \forall v\in V
\end{array}
$}
\end{align}
The dual separation oracle can be solved using \WSSE; so we can apply the Ellipsoid algorithm. Since we only have a
multi-criteria approximation for \WSSE (see Section~\ref{sec:wsse}),
the details for approximating the configuration LP are
rather technical.

\subsection{Aggregation}\label{sec:Aggregation}
The aggregation process, which might be of independent interest, transforms a cover of $G$ into a partition.
Intuitively, we first let the sets randomly compete with each other over
the vertices so as to form a partition;
then, to make sure no set has large cost,
we repeatedly fix the partition locally,
and use a potential function to track progress.

\begin{theorem}\label{thm:aggreg}
Algorithm \ref{alg:RoundMMKP} is a randomized polynomial-time algorithm
that when given a graph $G=(V,E)$,
an $\varepsilon\in (0,1)$, and a cover $\calS$ of $V$ that satisfies:
(i) every vertex in $V$ is covered
by at least $c/k$ fraction of sets $S\in \calS$, for $c\in(0,1]$; and
(ii) all $S\in \calS$ satisfy $|S|\leq 2n/k$ and $\delta(S) \leq B$;
the algorithm outputs a partition $\calP$ of $V$ into at most $k$ sets such that
for all $P\in \calP$ we have $|P|\leq 2(1+\varepsilon) n/k$ and
$\E[\max{\delta(P): P\in \calP}] \leq 8B/(c\varepsilon)$.
\end{theorem}

\LinesNotNumbered\DontPrintSemicolon
\begin{algorithm}[!ht]
\caption{\label{alg:RoundMMKP} Aggregation Procedure for \MMKP:}
\nlset{1}\label{alg:Sampling}\tct{Sampling}\Indp
Sort sets in $\calS$ in a random order: $S_1, S_2,\dots, S_{|\calS|}$.
Let $P_i=S_i\setminus \cup_{j < i}S_j$.\;
\Indm\nlset{2}\label{alg:Uncrossing}\tct{Replacing Expanding Sets with Sets from $\calS$}\Indp%
\While{there is a set $P_i$ such that $\delta(P_i)> 2B$}
{
 Set $P_i = S_i$,
 and for all $j\neq i$, set $P_j = P_j\setminus S_i$.
}
\Indm\nlset{3}\label{alg:Aggregating}\tct{Aggregating}
Let $B'=\max\{\frac{1}{k}\sum_{i}\delta(P), 2B\}$.\;
\While{there are $P_i\neq \varnothing$, $P_j\neq \varnothing$ ($i\neq j$) such that
$|P_i|+|P_j|\leq 2(1+\varepsilon)n/k$ and $\delta(P_i)+\delta(P_j)\leq 2B'\varepsilon^{-1}$}
{
 Set $P_i = P_i \cup P_j$ and set $P_j=\varnothing$.
}
\BlankLine
\nlset{4}\Return all non-empty sets $P_i$.
\end{algorithm}

\medskip

\begin{proof}[\textbf{\em{Analysis.}}]
1. Observe that after step 1 the collection of sets $\{P_i\}$ is a partition of $V$ and $P_i\subset S_i$ for  every
$i$. Particularly, $|P_i|\leq |S_i|\leq 2 n/k$. Note, however, that the bound $\delta(P_i)\leq B$ may be violated for
some $i$. We now prove that $\E\bigl[\sum_{i}\delta(P_i)\bigr]\leq 2kB/c$. Fix an $i\leq |\calS|$ and estimate the expected
weight of edges $E(P_i,\cup_{j> i} P_j)$ given that $S_i = S$. If an edge $(u,v)$ belongs to $E(P_i,\cup_{j> i} P_j)$
then $(u,v) \in  E(S_i, V\setminus S_i) = E(S, V\setminus S)$ and both $u, v\notin \cup_{j< i} S_j$. For any edge
$(u,v)\in \delta(S)$ (with $u\in S$, $v\notin S$), $\Pr((u,v)\in E(P_i,\cup_{j> i} P_j) \given S_i=S) \le \Pr(v\notin
\cup_{j< i} S_j\given S_i=S) \leq (1- c/k)^{i-1}$, since $v$ is covered by at least $c/k$ fraction of sets in $\calS$
and is not covered by $S_i=S$. Hence,
$$\E[w(E(P_i,\cup_{j> i} P_j))\given S_i = S] \leq (1- c/k)^{i-1}\delta(S)\leq (1-c/k)^{i-1} B,$$
and $\E[w(E(P_i,\cup_{j> i} P_j))] \leq (1-c/k)^{i-1} B$. Therefore, the total expected weight of edges crossing
the boundary of $P_i$'s is at most $\sum_{i=0}^{\infty} (1-c/k)^i B = kB/c$, and
$\E\bigl[\sum_{i}\delta(P_i)\bigr]\leq 2kB/c$.

 2. After each iteration of step 2, the following invariant holds: the collection of sets $\{P_i\}$ is
a partition of $V$ and $P_i\subset S_i$ for all $i$. Particularly, $|P_i|\leq |S_i|\leq 2 n/k$. The key observation is
that at every iteration of the ``while'' loop, the sum $\sum_{j}\delta(P_j)$ decreases by at least $2B$. This is due to
the following uncrossing argument:
\begin{eqnarray*}
\delta(S_i) + \sum_{j\neq i}\delta(P_j\setminus S_i)&\leq& \delta(S_i) + \sum_{j\neq i} \Bigl(\delta(P_j) +
w(E(P_j\setminus S_i, S_i)) - w(E(S_i\setminus P_j,P_j) \Bigr) \\
&\leq& \delta(S_i) + \Bigl(\sum_{j\neq i}\delta(P_j)\Bigr) + \underbrace{w(E(V\setminus S_i, S_i))}_{\delta(S_i)} - \underbrace{w(E(P_i,V\setminus P_i))}_{\delta(P_i)}\\
&=& \Bigl(\sum_{j}\delta(P_j)\Bigr) + 2\delta (S_i) - 2\delta (P_i) \leq \Bigl(\sum_{j}\delta(P_j)\Bigr) - 2B.
\end{eqnarray*}
we used that $P_i\subset S_i$, all $P_j$ are disjoint, $\cup_{j\neq i}(P_j\setminus S_i)\subset V\setminus S_i$, $P_i
\subset S_i\setminus P_j$, $\cup_{j\neq i}P_j =  V\setminus P_i$. Hence, the number of iterations of the loop in step 2
is always polynomially bounded and after the last iteration $\E\bigl[\sum_{i}\delta(P_i)\bigr]\leq 2kB/c$ (the
expectation is over random choices at step 1; the step 2 does not use random bits). Hence, $\E[B']\leq 4B/c$.

3. The following analysis holds conditional on any value of $B'$. After each iteration of step 3, the following
invariant holds: the collection of sets $\{P_i\}$ is a partition of $V$. Moreover, $|P_i|\leq 2(1+\varepsilon) n/k$ and
$\delta(P_i)\leq 2B'\varepsilon^{-1}$ (note: after step 2, $\delta(P_i)\leq 2B\leq B'$ for each $i$).

When the loop terminates, we obtain a partition of $V$ into sets $P_i$ satisfying
$|P_i|\leq 2(1+\varepsilon) n/k$, $\sum_i |P_i| = n$, $\delta(P_i)\leq 2B'\varepsilon^{-1}$,
$\sum_i \delta (P_i)\leq k B'$, such that no two sets can be merged without violating above
constraints. Hence by Lemma~\ref{lem:AggLemma} below (with $a_i=|P_i|$ and $b_i=\delta(P_i)$), the number of non-empty sets is at most
$2\;\frac{n}{2(1+\varepsilon)n/k}+\frac{kB'}{2B'\varepsilon^{-1}} =
(1+\varepsilon)^{-1}k + (\varepsilon/2)k \leq k.$
\end{proof}

\begin{lemma}[Greedy Aggregation] \label{lem:AggLemma}
Let $a_1,\dots, a_t$ and $b_1,\dots b_t$ be two sequences of nonnegative numbers satisfying
the following constraints $a_i< A$, $b_i <B$, $\sum_{i=1}^t a_i \leq S$ and $\sum_{i=1}^t b_i \leq T$ (for
some positive real numbers $A$, $B$, $S$, and $T$).
Moreover, assume that for every $i$ and $j$ ($i\neq j$) either $a_i+a_j > A$ or $b_i + b_j > B$.
Then, $t < S/A + T/B + \max (S/A, T/B, 1)$.
\end{lemma}
\begin{proof}
By rescaling we assume that $A=1$ and $B=1$. Moreover, we may assume that $\sum_{i=1}^t a_i < S$
and $\sum_{i=1}^t b_i < T$ by slightly decreasing values of all $a_i$ and $b_i$ so that all
inequalities still hold.

We write two linear programs. The first LP ($LP_I$) has variables $x_i$ and constraints $x_i + x_j \geq 1$
for all $i,j$ such that $a_i+a_j\geq 1$. The second LP ($LP_{II}$) has variables $y_i$ and constraints $y_i + y_j \geq 1$
 all $i,j$ such that $b_i+b_j\geq 1$. The LP objectives are to minimize $\sum_i x_i$ and to minimize $\sum_i y_i$. Note, that
$\{a_i\}$ is a feasible point for $LP_I$ and $\{b_i\}$ is a feasible point for $LP_{II}$. Thus,
the optimum values of $LP_I$ and $LP_{II}$ are strictly less than $S$ and $T$ respectively.

Observe that both LPs are half-integral. Consider optimal solutions $x^*_i$, $y^*_j$ where
$x^*_i,y^*_j\in \{0,1/2,1\}$. Note that for every $i,j$ either $x^*_i+x^*_j \geq 1$ or
$y^*_i+y^*_j \geq 1$. Consider several cases. If for all $i$, $x^*_i+ y^*_i\geq 1$, then $t< S+T$, since $\sum_{i=1}^t(x^*_i+y^*_i) < S+T$. If for some $j$, $x^*_j+ y^*_j = 0$ (and hence $x^*_j=y^*_j=0$), then
$x^*_i+ y^*_i\geq 1$ for $i\neq j$ and, thus, $t < S+T+1$. Finally, assume that
for some $j$, $x^*_j+y^*_j = 1/2$, and w.l.o.g. $x^*_j=1/2$ and $y^*_j=0$. The number of $i$'s
with $x^*_i\neq 0$ is (strictly) bounded by $2S$. For the remaining $i$'s, $x^*_i=0$ and hence $y^*_i=1$
(because $y^*_i = y^*_i+y^*_j\geq 1$), and thus the number of such $i$'s is (strictly) bounded by $T$.
\end{proof}

\section{Further Extensions}\label{sec:FE}
Both Theorems \ref{thm:bal} and \ref{thm:mcut} follow from a more general result for a problem that we call \mmc,
defined as follows. The input is an undirected graph $G=(V,E)$, nonnegative edge-weights $w$, a collection of disjoint
terminal sets $T_1,T_2,\ldots,T_k \subset V$ (possibly empty), and parameters $\rho \in [1/k,1]$ and $C,D>0$. The goal
is to find a partition $S_1,\ldots,S_k$ of $V$ such that:
\begin{enumerate} \compactify
\item For all $i$, $\ T_i \subseteq S_i$;
\item For all $i$, $\ |S_i| \leq \rho n$;
\item For all $i$, $\ \delta(S_i) \leq C$; and
\item $\sum_i \delta(S_i) \leq D$.
\end{enumerate}

This problem models the aforementioned cloud computing scenario, where in addition, certain processes are preassigned
to machines (each set $T_i$ maps to machine $i\in[k]$). The goal is to assign the processes $V$ to machines $[k]$ while
respecting the preassignment and machine load constraints, and minimizing both bandwidth per machine and total volume
of communication.

\begin{theorem}
\label{thm:gen} There is a randomized polynomial time algorithm that given any feasible instance of the \mmc problem
with parameters $k,\rho,C,D$ and any $\varepsilon>0$, finds a partition $Q_1,\ldots,Q_k$ with the following properties:
(i) For all $i$, $\ T_i \subseteq Q_i$; (ii) For all $i$, $\ |Q_i| \leq (2+ \varepsilon) \rho n$; (iii) $\E\left[
\max_{i=1}^k \delta(Q_i) \right] \leq O_{\varepsilon}(\sqrt{\log n \log k}) C$; and (iv) $\E\left[ \sum_i \delta(Q_i)
\right] \leq O_{\varepsilon}(\sqrt{\log n \log k}) D$.
\end{theorem}

It is clear that in fact Theorem~\ref{thm:gen} generalizes both Theorems~\ref{thm:bal} and \ref{thm:mcut}. Let us now
describe modifications to the \MMKP algorithm used to obtain Theorem~\ref{thm:gen}.

\paragraph{Uniform Coverings.} First, by the introduction of vertex weights, we can shrink each preassigned set $T_i$ to a single terminal $t_i$ (for
$i\in[k]$). Then, feasible vertex-sets $\cs$ in the covering procedure (Section~\ref{sec:mult}) consist of those $S\sse
V$ where $\mbox{weight}(S)\le \rho\,n$ (balance constraint) and $|S\cap \{t_i\}_{i=1}^k|\le 1$ (preassignment
constraint). The subproblem \MBWC also has the additional $|S\cap \{t_i\}_{i=1}^k|\le 1$ constraint; this can be
handled in the algorithm from Section~\ref{sec:wsse} by guessing which terminal belongs to $S$ (see
Remark~\ref{rem:terminals}). Using Corollary~\ref{cor:sdp-subroutine} we assume an $(\alpha,\beta,\gamma)$
approximation algorithm for this (modified) \MBWC problem; where for any $\varepsilon>0$,
$\alpha=O_\varepsilon(\sqrt{\log n \, \log(\max\{1/\rho,1/\tau\})})$, $\beta=1+\varepsilon$ and
$\gamma=O_\varepsilon(1)$.

Algorithm~\ref{alg:covMMC-gen} below gives the procedure to obtain a uniform covering $\mathcal{S}$ bounding total
edge-cost in addition to the conditions in Theorem~\ref{thm:covering}.

\begin{algorithm}[!ht]
\caption{Covering Procedure for \mmc:} \label{alg:covMMC-gen}
set $t\gets 1$, $y^1(v)\gets 1$ for all $v\in V$, and $Y^1\gets \sum_{v\in V} y^1(v)$.\\
 \While{$Y^t > \frac1n$}{
 \For{$i=0,\ldots,\log_2 k+1$}{
solve the \MBWC instance $\langle G, \, y^t,\, w,\, \frac{1}{2^i},\, \rho\rangle$
using the algorithm from Corollary~\ref{cor:sdp-subroutine}, to obtain $S^t(i)\sse V$.\\
If $\left(\delta(S^t(i))\le \alpha\cdot \min\{ C,\, 4D/2^i\}\right)$ then $S^t\gets S^t(i)$ and quit for loop.
 }
 set $\calS\gets \calS\cup \{S^t\}$.\\
  \For{$v\in V$}{ set $y^{t+1}(v)\gets \frac12\cdot y^{t}(v)$ if $v\in S^t$, and
$y^{t+1}(v)\gets y^{t}(v)$  otherwise. }
set $Y^{t+1}\gets \sum_{v\in V} y^{t+1}(v)$.\\
set $t\gets t+1$.
 }
output $\calS$.
\end{algorithm}

\begin{theorem}\label{thm:covering-general}
For any instance of \mmc, output $\calS$ of Algorithm~\ref{alg:covMMC-gen} satisfies:
\begin{enumerate}
 \item $\delta(S)\le \alpha\cdot C$ and $|S|\le \beta\cdot \frac{n}k$ for all $S\in \calS$.
 \item $|\{S\in \calS : S\ni v\}| \ge \log_2 n$ for all $v\in V$.
 \item $|\calS|\le 5\gamma \, k\cdot \log_2 n$.
 \item $\sum_{S\in \calS} \delta(S) \le 17\alpha\,\gamma\, \log_2n\cdot D$.
\end{enumerate}
Above, for any $\varepsilon >0$, $\alpha=O_\varepsilon(\sqrt{\log n \, \log k})$, $\beta=1+\varepsilon$ and
$\gamma=O_\varepsilon(1)$.
\end{theorem}
\begin{proof}
We start with the following key claim.
\begin{claim} \label{cl:mult-update-general}
In any iteration $t$ of the above algorithm, there exists an $i\in \{0,1,\ldots,\log_2k+1\}$ such that
$\delta(S^t(i))\le \alpha\cdot \min\{ C,\, 4D/2^i\}$, $|S^t(i)|\le \beta\cdot \frac{n}k$, and $y^t(S^t(i))\ge
\frac{Y^t}{\gamma\,2^i}$.
\end{claim}
\begin{proof}
Consider the optimal solution $\{S^*_j\}_{j=1}^k$ of the original \mmc instance.  For all $j\in [k]$ we have that
$|S^*_j|\le \rho n$, $\delta(S^*_j)\le C$ and $S^*_j$ contains at most one terminal. Moreover, $\sum_{j=1}^k
\delta(S^*_j)\le D$. Since $\{S^*_j\}_{j=1}^k$ partitions $V$, we also have $\sum_{j=1}^k y^t(S^*_j)= Y^t$. Let
$L\sse[k]$ denote the indices $j$ having $\delta(S^*_j)\le \frac{2D}{Y^t}\cdot y^t(S^*_j)$.

We claim that $\sum_{j\in L} y^t(S^*_j) \ge Y^t/2$. This is because:
$$\sum_{j\not\in L} y^t(S^*_j) \le \frac{Y^t}{2D}
\cdot \sum_{j\not\in L} \delta(S^*_j) \le \frac{Y^t}{2D} \cdot \sum_{j=1}^k \delta(S^*_j) \le Y^t/2.$$

Since $|L|\le k$, there is some $q\in L$ with $y^t(S^*_q)\ge \frac{Y^t}{2k}$. Let $i\in \{1,\ldots,\log_2k+1\}$ be the
value such that $y^t(S^*_q)/ Y^t \in [\frac{1}{2^{i}},\, \frac{1}{2^{i-1}}]$; note that such an $i$ exists because
$y^t(S^*_q)/ Y^t \in [\frac1{2k},1]$. For this $i$, consider the \MBWC instance $\langle G, \, y^t,\, w,\,
\frac{1}{2^i},\, \rho\rangle$. Observe that $S^*_q$ is a feasible solution here since $y^t(S^*_q) \ge Y^t/2^i$,
$|S^*_q|\le \rho n$ and $S^*_q$ contains at most one terminal. Hence the optimal value of this instance is at most:
$$\delta(S^*_q)\le \min\left\{ C,\,\frac{2D}{Y^t}\cdot y^t(S^*_q)\right\} \le \min\left\{ C,\,\frac{4D}{2^i}\right\}$$
The first inequality uses the definition of $L$ and that $\delta(S^*_q)\le C$, and the second inequality is by choice
of $i$. It now follows from Corollary~\ref{cor:sdp-subroutine} that solution $S^t(i)$ satisfies the claimed properties.
We note that $\alpha=O_\varepsilon(\sqrt{\log n \, \log k})$ because each instance of \MBWC has parameters
$\tau=\frac1{2^i}\ge \frac1k$ and $\rho\ge \frac1k$.
\end{proof}

For each iteration $t$, let $i_t\in \{1,\ldots,\log_2k+1\}$ be the index such that $S^t = S^t(i_t)$.
Claim~\ref{cl:mult-update-general} implies that such an index always exists, and so the algorithm is indeed
well-defined. {\bf Condition~1} of Theorem~\ref{thm:covering-general} also follows directly. Let $\ell$ denote the
number of iterations of the while loop, for the given \mmc instance. For any $v\in V$, let $N_v$ denote the number of
iterations $t$ with $S^t\ni v$. Then, by the $y$-updates we have $y^{\ell+1}(v)=1/2^{N_v}$. Moreover, the termination
condition implies that $y^{\ell+1}(v)\le \frac1n$ (since $Y^{\ell+1}\le \frac1n$). Thus we obtain $N_v\ge \log_2 n$ for
all $v\in V$, proving {\bf condition~2} of Theorem~\ref{thm:covering-general}.

Claim~\ref{cl:mult-update-general} implies that for each iteration $t$, we have $y^t(S^t)\ge \frac{Y^t}{\gamma\,
2^{i_t}}$ and $\delta(S^t)\le 4\alpha\, D/2^{i_t}$. Since $i_t\le \log_2k+1$, we obtain:
$$y^t(S^t) \ge \max\left\{ \frac{Y^t}{2\gamma\,k},\, \frac{Y^t}{4\alpha \gamma\,D}\cdot \delta(S^t)\right\}.$$

\begin{enumerate}
 \item Using $y^t(S^t) \ge \frac{Y^t}{2\gamma\,k}$ in each iteration,
 $Y^{t+1}=Y^t-\frac12 \cdot y^t(S^t)\le \left( 1-\frac1{4\gamma\, k} \right)\cdot Y^t$. This implies that $Y^{\ell}\le
\left( 1-\frac1{4\gamma\, k} \right)^{\ell-1} \cdot Y^1= \left( 1-\frac1{4\gamma\, k} \right)^{\ell-1} \cdot n$.
However $Y^\ell>\frac1n$ since the algorithm performs $\ell$ iterations. Thus, $\ell\le 1+8\gamma\,k\cdot \ln n\le
9\gamma\, k\cdot \log_2n$. This proves {\bf condition 3} in Theorem~\ref{thm:covering-general}.
 \item Using $y^t(S^t) \ge \frac{Y^t}{4\alpha \gamma\,D}\cdot \delta(S^t)$ in each iteration,
  $Y^{t+1}=Y^t-\frac12 \cdot y^t(S^t)\le \left( 1-\frac{\delta(S^t)}{8\alpha \gamma\,D} \right)\cdot Y^t$. So,
  $$\frac1n < Y^{\ell}\le \Pi_{t=1}^{\ell-1} \left( 1-\frac{\delta(S^t)}{8\alpha \gamma\,D} \right) \cdot Y^1\le
  \exp\left( -\frac{\sum_{t=1}^{\ell-1} \delta(S^t)}{8\alpha \gamma\,D} \right)\cdot n$$
  This implies $\sum_{t=1}^{\ell-1} \delta(S^t) \le (16\alpha\gamma\, \ln n)\cdot D$. Adding in $\delta(S^\ell)\le
  \alpha\cdot C\le \alpha\, D$, we obtain {\bf condition~4} of Theorem~\ref{thm:covering-general}.
\end{enumerate}
This completes the proof of Theorem~\ref{thm:covering-general}.
\end{proof}


\paragraph{Aggregation} This step remains essentially the same as in Section~\ref{sec:Aggregation}, namely
Algorithm~\ref{alg:Aggregating} (with parameter $B:=\alpha\cdot C$). The only difference is that in Step~3 we do not
merge parts containing terminals. We first show that this yields a slightly weaker version of Theorem~\ref{thm:gen}: in
condition (ii) we obtain a bound of $(3+\varepsilon)\rho n$ on the cardinality of each part. (Later we show how to
achieve the cardinality bound of $(2+\varepsilon)\rho n$ as claimed in Theorem~\ref{thm:gen}.)

Note that each of the final sets $\{P_i\}$ is a subset of some set in ${\cal S}$, and hence contains at most one
terminal. It also follows that the final sets $\{P_i\}$ are at most $2k$ in number: at most $k$ of them contain no
terminals (just as in Theorem~\ref{thm:aggreg}), and at most $k$ contain a terminal (since there are at
most $k$ terminals). Each of these sets $\{P_i\}$ has size at most $(2+\varepsilon)\rho n$ and cut value at most
$8B/(c\varepsilon)$, by the analysis in Theorem~\ref{thm:gen}. Moreover, if a set $P_i$ contains a terminal then
$|P_i|\le \beta\cdot \rho n= (1+\varepsilon)\rho n$ (since it does not participate in any merge). Finally in order to
reduce the number of parts to $k$, we merge arbitrarily each part containing a terminal with one non-terminal part; and
output this as the final solution. It is clear that each part has at most one terminal, has size $\le
(3+\varepsilon)\rho n$, and cut value at most $O_\varepsilon(\sqrt{\log n \log k})\cdot C$. The bound on total cost
(condition (iv) in Theorem~\ref{thm:gen}) is by the following claim. This proves a weaker version of
Theorem~\ref{thm:gen}, with size bound $(3+\varepsilon)\rho n$.
\begin{claim}\label{cl:gen-sum-obj}
Algorithm~\ref{alg:Aggregating} applied on collection $\calS$ from Theorem~\ref{thm:covering-general} outputs partition
$\{P_i\}_{i=1}^k$ satisfying $\E\left[\sum_{i=1}^k w(\delta(P_i))\right] = O_{\varepsilon}(\sqrt{\log n\,\log k})\, D$.
\end{claim}
\begin{proof}
We will show that the random partition $\{P_i\}$ at the end of Step~1 in Algorithm~\ref{alg:Aggregating} satisfies
$\E\left[\sum_{i} w(\delta(P_i))\right] \le  O_{\varepsilon}(\sqrt{\log n\,\log k})\, D$. This would suffice since
$\sum_{i} w(\delta(P_i))$ does not increase in Steps 2 and 3. For notational convenience, we assume (by adding empty
sets) that $|\calS|=5\gamma\,k \cdot \log_2n$ in Theorem~\ref{thm:covering-general}; note that this does not affect any
of the other conditions.

To bound the cost of the partition $\{P_i\}$ in Step~1, consider any index $i\le |\calS|$. From the proof of
Theorem~\ref{thm:aggreg}, we have:
$$\E[w(E(P_i,\cup_{j> i} P_j))\given S_i = S] \leq (1- c/k)^{i-1}\delta(S),$$
where $c=\frac1{5\gamma}$ is such that each vertex lies in at least $c/k$ fraction of sets $\calS$. Deconditioning,
$\E[w(E(P_i,\cup_{j> i} P_j))] \leq (1- c/k)^{i-1} \cdot \E[\delta(S_i)]= (1- c/k)^{i-1} \cdot \left( \sum_{S\in\calS}
\delta(S) / |\calS|\right)$, where we used that $S_i$ is a uniformly random set from $\calS$. So the total edge-cost,
$$\E\left[ \sum_i \delta(P_i)\right] = 2\cdot \sum_i \E[w(E(P_i,\cup_{j> i} P_j))] \leq
\left( \sum_{i\ge 0} (1- c/k)^{i}\right) \cdot \left( \sum_{S\in\calS} \delta(S) / |\calS|\right) = \frac{k}c\cdot
\sum_{S\in\calS} \delta(S) / |\calS|.$$
 Using $\sum_{S\in\calS} \delta(S)\le 17\alpha\gamma\log_2n\cdot D$ and
$|\calS|=5\gamma\,k \cdot \log_2n$ from Theorem~\ref{thm:covering-general}, $\E\left[ \sum_i \delta(P_i)\right] \le
\frac{17\alpha}{5c}\, D = O_{\varepsilon}(\sqrt{\log n\,\log k})\, D$ since $\alpha=O_{\varepsilon}(\sqrt{\log n\,\log
k})$ and $1/c=O_{\varepsilon}(1)$.
\end{proof}

\medskip
\noindent {\bf Obtaining size bound of $(2+\varepsilon)\rho n$.} We now describe a modified aggregating step (in place
of Step~3 in Algorithm~\ref{alg:Aggregating}) that yields Theorem~\ref{thm:gen}. Given the uniform cover ${\cal S}$
from Algorithm~\ref{alg:covMMC-gen}, run Steps 1 and 2 of Algorithm~\ref{alg:Aggregating} (use $B=\alpha C$) to obtain
parts $P_1,\ldots,P_{|{\cal S}|}$. Then:
\begin{enumerate}
\item Set $B' := \max\left\{ \frac1k \sum_i \delta(P_i),\, 2B\right\}$.
\item While there are $P_i,P_j\ne \varnothing$ ($i\ne j$) such that $|P_i|+|P_j|\le (1+\varepsilon)\rho n$,
$\delta(P_i)+\delta(P_j)\le 2B'$ and $P_i\cup P_j$ does not contain a terminal: replace $P_i\gets P_i\cup P_j$ and
$P_j\gets \varnothing$.
\item Sort the resulting non-empty sets $P_1,\ldots,P_t$ in non-increasing order of size.
\item Form $\lceil t/k \rceil$ groups where the $j^{th}$ group consists of  parts indexed between $(j-1)k+1$ and
$jk$.
\item For each $i\in [k]$ define $Q_i$ as the union of one part from each group such that it contains terminal $i$ but
no other terminal. Additionally, ensure that each part is assigned to one of $\{Q_i\}_{i=1}^k$.
\end{enumerate}

We first show that the number of parts after Step~2 above $t\le 4k$. Note that each part contains at most one terminal,
and the number of parts containing a terminal is at most $k$. For the non-terminal parts, using
Lemma~\ref{lem:AggLemma} (with $a_i=|P_i|$, $b_i=\delta(P_i)$, $a=(1+\varepsilon)\rho n$, $b=2B'$, $S=n$ and $T=kB'$)
we obtain a bound of $5k/2$, which implies $t\le 4k$.

Next observe that sets $\{Q_i\}_{i=1}^k$ in Step~5 above are well-defined, and can be found using a simple greedy rule:
this is because each group contains at most $k$ parts. This gives Theorem~\ref{thm:gen} (i). Since $t\le 4k$ there are
at most 4 groups and hence $\max_{i=1}^k \delta(Q_i)\le 4\cdot \max_{i=1}^t \delta(P_i) \le 8B'$. This proves
Theorem~\ref{thm:gen} (iii). Also, the proof of Claim~\ref{cl:gen-sum-obj} implies Theorem~\ref{thm:gen} (iv); this is
due to the fact that in the final partition $\sum_{i=1}^k \delta(Q_i) \le \sum_{\ell=1}^{|{\cal S}|} \delta(P_\ell)$.

We now show Theorem~\ref{thm:gen} (ii), i.e. $\max_{i=1}^k |Q_i| \le (2+\varepsilon)\rho n$. Consider any $i\in [k]$
and let $P'_i$ denote the part assigned to $Q_i$ from the first group. By the non-increasing size ordering
$P_1,\ldots,P_t$ and the round-robin assignment (Step~5) into $Q_i$s, we obtain that $|Q_i|-|P'_i|\le |Q_j|$ for all
$j\in [k]$. Taking an average, $|Q_i|-|P'_i|\le \frac1k \sum_{j=1}^k |Q_j|=\frac{n}k$. Finally, since each part
$\{P_\ell\}_{\ell=1}^t$ has size at most $(1+\epsilon)\rho n$ we have $|Q_i|\le (2+\epsilon)\rho n$. This completes the
proof of Theorem~\ref{thm:gen}.


\section{Hardness of \MMMC} \label{sec:Hardness}
\def\a{\ensuremath{\mathcal{A}}\xspace}
\def\q{\ensuremath{\mathcal{Q}}\xspace}
\def\is{\ensuremath{\mathcal{I}}\xspace}
\def\js{\ensuremath{\mathcal{J}}\xspace}

In this section we prove Theorem \ref{thm:multiway-hard},
which shows that obtaining a $k^{1-\varepsilon}$-approximation algorithm
for \MMMC is hard, if not unlikely. This suggests that some
dependence on $n$ might be necessary (unless we are satisfied
with approximation that is linear in $k$, which is trivial),
which is in contrast to several cut problems (multiway-cut, multicut,
requirement-cut etc) with sum-objective where $poly(\log k)$ approximation guarantees are known when $k$ is the size of
the terminal set~\cite{M09,LM10,EGKRTT10,CLLM10,MM10}.
Throughout this section, we assume $k$ is constant (independent of $n$),
which simplifies the statements;
strictly speaking, our reductions relate solving one problem
with parameter $k(n)$ to solving another problem with parameter $k(n')$.

We will refer to the min-sum version of \MMKP, called \MKP,
in which the input is an edge-weighted graph $G=(V,E)$ and a parameter $k$,
and the goal is to partition the vertices into $k$ equal-sized parts
while minimizing the total edge-weight of all edges cut.
An algorithm for \MKP is an $(\alpha,\beta)$ bicriteria approximation
if for every instance, it partitions the vertices into $k$ pieces,
each of size at most $\beta|V|/k$,
and the total edge-weight of all edges cut is at most $\alpha$ times
the least possible among all partitions into $k$ equal-size sets.

The basic idea in proving Theorem \ref{thm:multiway-hard} as follows.
although there is no vertex-balance requirement in \MMMC,
an edge-balance is implicit in the
objective. By introducing a complete bipartite graph (having suitable edge weight) between the terminals and the rest
of the graph, this edge-balance can be used to enforce the vertex-balance required in \MKP.
After this first step, which obtains a bicriteria approximation for \MKP,
we do a second step which improves the size-violation in the algorithm
for \MKP.
These two steps are formalized in the two foregoing Lemmas
\ref{lem:fromMMMC2MKP} and \ref{lem:fromMKP2MKP}.
Putting them together immediately proves \ref{thm:multiway-hard}.

\begin{lemma}\label{lem:fromMMMC2MKP}
If there is a $\rho$-approximation algorithm for \MMMC then there is a $(5\rho\,k, \,10\rho)$-bicriteria approximation
algorithm for \MKP.
\end{lemma}
\begin{proof}
Let $\a$ denote the $\rho$-approximation algorithm for \MMMC. Consider any instance \is of \MKP: graph $G=(V,E)$, edge
costs $c:E\rightarrow \mathbb{R}_+$ and parameter $k$. For every $B\ge 0$, consider an instance $\js(B)$ of \MMMC on
graph $G_B$ defined as:
\begin{OneLiners}
\item The vertex set is $V\bigcup \{t_i\}_{i=1}^k$ where $\{t_i\}_{i=1}^k$ are the terminals.
\item The edges are $E\bigcup \{(t_i,u) : i\in[k],\, u\in V\}$.
\item Extend cost function $c$ by setting $c(t_i,u)=\frac{B}n$ for all $i\in[k],\, u\in V$.
\end{OneLiners}
Note that every solution to $\js(B)$ corresponds to a $k$--partition in $G$ (though possibly unbalanced). We say that a
solution to $\js(B)$ is $\beta$-balanced if each piece in the partition has size at most $\beta\cdot \frac{n}k$.

The algorithm for \MKP on \is runs algorithm \a on all the \MMMC instances $\left\{\js(2^i) : 0\le i\le \log_2
\left(\sum_{e\in E} c_e \right) \right\}$, and returns the cheapest partition that is $(10\rho)$-balanced. We now show
that this results in a $(5\rho\,k,\,10\rho)$ bicriteria approximation ratio.

Note that algorithm \a must be invoked on $\js(B)$ for some value $B$ with $B<\opt(\is)\le 2\cdot B$. We will show that
the partition resulting from this call is the desired bicriteria approximation.

\begin{claim}\label{cl:phardness-1}
$\opt(\js(B))\le \opt(\is)+2B\le 5\opt(\is)$.
\end{claim}
\begin{proof}  Let $P^*$ denote the optimal $k$--partition to \is.
Consider the solution to $\js(B)$ obtained by including each terminal into a distinct piece of $P^*$. The boundary of
the piece containing $t_i$ (any $i\in[k]$) in $G_B$ costs at most:
$$\opt(\is) + n\cdot \frac{B}n + \frac{n}k\cdot (k-1)\cdot \frac{B}n\le \opt(\is) + 2B,$$
the term $\opt(\is)$ is due to edges in $E$, the second term is due to edges at $t_i$ and the third is due to edges at
all other $\{t_j:j\ne i\}$. The claim now follows since $B\le 2\,\opt(\is)$.
\end{proof}

Let $P$ denote $\a$'s solution to $\js(B)$, and $\{P_i\}_{i=1}^k$ the partition of $V$ induced by $P$. From
Claim~\ref{cl:phardness-1}, $P$ has objective value (for \MMMC) at most $5\rho\, \opt(\is)$. Note that the boundary (in
graph $G_B$) of $t_i$'s piece  in $P$ costs {\em at least} $c(\delta_G(P_i))+(k-1)\frac{B}n\cdot |P_i|$. Thus we
obtain:
$$c(\delta_G(P_i))+(k-1)\frac{B}n\cdot |P_i| \le 5\rho\, \opt(\is),\qquad \mbox{ for all }i\in[k].$$
It follows that for every $i\in[k]$, we have:
\begin{enumerate}
\item $|P_i|\le \frac{5\rho\,\opt(\is)}{B}\cdot \frac{n}{k-1}\le 10\rho\cdot \frac{n}k$ since $B>\opt(\is)$,  and
\item $c(\delta_G(P_i)) \le 5\rho\cdot\opt(\is)$.
\end{enumerate}
Thus the solution $\{P_i\}_{i=1}^k$ to \is is $(10\rho)$-balanced and costs at most $5\rho\,k\cdot \opt(\is)$.
This complete the proof of Lemma \ref{lem:fromMMMC2MKP}.
\end{proof}

\begin{lemma}\label{lem:fromMKP2MKP}
If there is an $(\alpha, \,k^{1-\varepsilon})$-bicriteria approximation algorithm for \MKP for some constant $\varepsilon>0$,
then there is also an $(\alpha\cdot\log\log k, \,\gamma)$-bicriteria
approximation algorithm with $\gamma\le 3^{2/\varepsilon}$.
\end{lemma}
\begin{proof}
Let $\a$ denote the $(\alpha, \,k^{1-\varepsilon})$-bicriteria approximation algorithm. The idea is to use \a
recursively to obtain the claimed $(\alpha\cdot\log\log k, \,\gamma)$ approximation; details are below. Let $G$ denote
the input graph with $n$ vertices, and \opt the optimal balanced $k$--partition. The algorithm deals with several
sub-instances of $G$, each of which is assigned a {\em level} from $\{0,1,\ldots, t\}$ where $t:=\Theta(\log\log k)$ is
fixed below. Every level $i$ instance will contain $k_i=k^{(1-\varepsilon/2)^i}$ parts and $n_i=\frac{n}k\cdot k_i$
vertices. Choose $t$ to be the smallest integer such that $k_t\le 3^{2/\varepsilon}$. While generating the sub-instances
we also add dummy singleton vertices (to keep the instances balanced). For notational simplicity we use the same
identifier for a sub-instance and the graph corresponding to it. Note that $G$ is the unique level $0$ instance. For
each $i\in \{0,1,\ldots,t-1\}$, every level $i$ instance \is generates $k_i$ level $i+1$ instances as follows:
\begin{itemize}
\item Run algorithm \a on \is to obtain a $k_i^{1-\varepsilon}$-balanced partition $\{P_j : 1\le j\le k_i\}$ of \is.
\item For each $j\in[k_i]$, add $n_{i+1}-|P_j|$ singleton vertices to $\is[P_j]$ to obtain a new level $i+1$
instance.
\end{itemize}
Note that this process is indeed well-defined since $|P_j|\le \frac{n_i}{k_i}\cdot
k_i^{1-\varepsilon}\le\frac{n_i}{k_i}\cdot k_i^{1-\varepsilon/2}=\frac{n_i}{k_i}\cdot k_{i+1}=n_{i+1}$. Moreover the
total number of level $i+1$ instances is $\Pi_{\ell=0}^{i} k_\ell =k^{\sum_{\ell=0}^i (1-\varepsilon/2)^\ell}\le
k^{2/\varepsilon}$. Since each instance has size at most $n$, the total size of all instances is polynomial.

The algorithm finally returns the partition \p corresponding to the set of all level $t$ instances. Note that there are
at most $n_t=\frac{n}k\cdot k_t\le 3^{2/\varepsilon}\cdot  \frac{n}k$ vertices in each level $t$ instance. The
algorithm ignores all dummy vertices in each piece of \p and greedily merges pieces until every piece has at least
$n/k$ vertices (all from $V$); let \p' denote the resulting partition of $V$. Clearly there are at most $k$ pieces in
\p', and each has size at most $3^{2/\varepsilon}\cdot  \frac{n}k$. Thus \p'  is a $3^{2/\varepsilon}$-balanced
$k$--partition.

We now upper bound the total cost of all edges removed by the algorithm (over all instances); this also bounds the cost
of \p'. This is immediate from Claim~\ref{cl:phardness-2} below: since there are $t$ levels and \a achieves an
$\alpha$-approximation to the cost, the total cost is bounded by $\alpha\, t\cdot \opt$.
\begin{claim}\label{cl:phardness-2}
For each $i\in\{1,\ldots,t\}$, the sum of optimal values of level $i$ instances is at most \opt.
\end{claim}
\begin{proof}
Consider any fixed $i$, and an instance \is of level $i$. We will show that the edges of \opt induced on \is form a
balanced $k_i$-partition for \is. This suffices to prove the claim since the level $i$ instances partition $V$ (the
original vertex-set). Let $V'$ denote the vertices from $V$ in \is; note that $|V'|\le \frac{n_{i-1}}{k_{i-1}}\cdot
k_{i-1}^{1-\varepsilon}$. Consider the partition \q of $V'$ induced by \opt; note that each piece in \q has size at
most $n/k$. Greedily merge pieces in \q as long as the size of each piece is at most $n/k$, to obtain
partition \q'; so the number of pieces is:
$$|\q'|\le 1+ 2|V'|\cdot \frac{k}n \le 1+ 2\frac{n_{i-1}}{k_{i-1}}\, k_{i-1}^{1-\varepsilon}\cdot \frac{k}n  =
1+2\,k_{i-1}^{1-\varepsilon}\le 3\,k_{i-1}^{1-\varepsilon}\le k_{i-1}^{1-\varepsilon/2}=k_i.$$ The second-last
inequality uses $k_{i-1}\ge 3^{2/\varepsilon}$ which is true by the choice of $t$. Now we can fill pieces of \q' with
dummy singleton vertices of instance \is to obtain a balanced $k_i$-partition of \is.
\end{proof}
This completes the proof of Lemma \ref{lem:fromMKP2MKP}.
\end{proof}

\bibliographystyle{alphainit}
\bibliography{cuts}

\section*{Appendix}
\appendix

\section{Integrality Gap for SDP Relaxation of \MMMC}
\label{sec:IntegralityGap} Consider the following semi-definite relaxation for the min-max multiway cut problem:
\begin{align}
\min ~~~~~ & \lambda & & \nonumber \\
s.t. ~~~~~ & \lambda \geq \sum _{(u,v)\in E} || y_{u,i} - y_{v,i}||_2^2 & & \forall i=1,2,\ldots , k \\
& \sum _{i=1}^k ||y_{u,i}||_2^2 = 1 & & \forall u\in V \label{constraint1}\\
& ||y_{t_i , i}||_2^2 = 1 & & \forall i=1,2,\ldots , k \label{constraint2}\\
& y_{u,i}\cdot y_{v,i} = 0 & & \forall u\in V, \forall i\neq j \label{constraint3}\\
& y_{u,i} \cdot \sum _{j=1}^k y_{v,j} = ||y_{u,i}||_2^2 & & \forall u,v\in V, \forall i=1,2,\ldots , k \label{constraint4}\\
& ||y_{u,i} - y_{v,j}||_2^2 + ||y_{v,j} - y_{w,r} ||_2^2 \geq ||y_{u,i} - y_{w,r} ||_2^2 & & \forall u,v,w\in V, \forall i,j,r = 1,2,\ldots , k \label{constraint5}\\
& y_{u,i}\cdot y_{v,j} \geq 0 & & \forall u,v\in V, \forall i,j=1,2,\ldots , k \label{constraint6}\\
& ||y_{u,i}||_2^2 \geq y_{u,i} \cdot y_{v,j} & & \forall u,v\in V, \forall u,v\in V, \forall i,j=1,2,\ldots , k
\label{constraint7}
\end{align}
In the above relaxation, each vertex $u\in V$ is associated with $k$ different vectors: $y_{u,1}, y_{u,2}, \ldots ,
y_{u,k}$. Each of those corresponds to a different terminal. Notice that the last two constraints of the relaxation are
the $\ell_2^2$ triangle inequality constraints including the origin (the same constraints used by \cite{CMM06}).

As stated, the integrality gap we consider is the star graph which constrains as single vertex $u$ connected by $k$
edges to the $k$ terminals. The value of any integral solution to this instance is exactly $k-1$ since the only choice
is to which terminal should $u$ be assigned. Any choice made results in a min-max objective value of $k-1$.

Let us construct the fractional solution to the above relaxation. Fix $e$ to be an arbitrary unit vector, and $x_1,
x_2, \ldots , x_k$ be $k$ unit vectors which are all orthogonal to $e$ and the inner product between any two of them is
$-1/(k-1)$. Set the following fractional solution:
\begin{align*}
y_{t_i , i} & = e & & \forall i=1,2, \ldots , k \\
y_{t_i , j} & = 0 & & \forall i\neq j \\
y_{u, i} & = \frac{1}{k}e + \frac{\sqrt{k-1}}{k} x_i & & \forall i=1,2,\ldots , k
\end{align*}

First, we show that the above fractional solution is feasible. Constraints (\ref{constraint1}), (\ref{constraint2}) and
(\ref{constraint3}) are obviously feasible for all terminals. Vertex $u$ also upholds constraint (\ref{constraint1})
since:
$$ \sum _{i=1}^k \left| \left| \frac{1}{k}e + \frac{\sqrt{k-1}}{k} x_i\right| \right| _2 ^2
 = \sum _{i=1}^k \left( \frac{1}{k^2} + \frac{k-1}{k^2}\right)  = 1 .$$
Let us verify that vertex $u$ satisfies also constraint (\ref{constraint3}):
$$ \left( \frac{1}{k}e + \frac{\sqrt{k-1}}{k} x_i\right) \cdot \left( \frac{1}{k}e + \frac{\sqrt{k-1}}{k} x_j\right)
= \frac{1}{k^2} + \frac{k-1}{k^2}\cdot \left( -\frac{1}{k-1}\right)  = 0 .$$

Focus on constraint (\ref{constraint4}). It is easy to verify that for all terminals, the sum of all $k$ vectors that
are associated with the picked vertex is exactly $e$. Since $\sum _{i=1}^k x_i = 0$, the sum of all vectors associated
with $u$ is also $e$. Therefore, all constraints of type (\ref{constraint4}) are satisfied.

Regarding the last three constraints, which are the $\ell_2^2$ triangle inequality constraints including the origin,
one can verify the following:
\begin{align*}
& ||e - 0||_2^2 = 1 \\
& \left| \left| e - \left( \frac{1}{k}e - \frac{\sqrt{k-1}}{k} x_i\right) \right| \right| _2^2 = 1-\frac{1}{k} \\
& \left| \left| 0 - \left( \frac{1}{k}e - \frac{\sqrt{k-1}}{k} x_i\right) \right| \right| _2^2 = \frac{1}{k} \\
& \left| \left| \left( \frac{1}{k}e - \frac{\sqrt{k-1}}{k} x_i\right) - \left( \frac{1}{k}e - \frac{\sqrt{k-1}}{k}
x_j\right) \right| \right| _2^2 = \frac{2}{k}
\end{align*}
All the last three constraints are derived from the above calculations. Hence, we can conclude that the fractional
solution defined above is feasible for the semi-definite relaxation.

The value of this solution is:
$$ \max _{1\leq i\leq k} \left\{ \left( 1-\frac{1}{k}\right) + (k-1)\cdot \frac{1}{k}\right\} = 2\frac{k-1}{k}.$$
This gives an integrality gap of $k/2$.


\section{\texorpdfstring{Bad Example: Greedy Algorithm for \MMKP}{Bad Example: Greedy Algorithm for Min--Max k-Partitioning}}\label{app:greedy-bad}
We show that the naive greedy algorithm that repeatedly uses \SSE to remove a part of size $n/k$ performs very poorly.
In this example we even assume that there is an exact algorithm for \SSE. The graph is a tree on  $n=k^2$ vertices $V:=
\{v\} \bigcup_{i=1}^{k-1} \{u_{i,0},\ldots,u_{i,k}\}$, and edges $E=\bigcup_{i=1}^{k-1} \{(u_{i,j},u_{i,j-1})\, :\,
1\le j\le k\}\, \bigcup \, \{(u_{i,0},u_{i-1,0})\, : \, 2\le i\le k-1 \} \bigcup \, \{(v,u_{1,0})\}$.

The simple greedy algorithm will cut out parts having small boundary for $k-1$ iterations, namely $P_i =
\{u_{i,1},\ldots,u_{i,k}\}$ for $i\in[k-1]$; note that $\delta(P_i)=1$ for all $i\in[k-1]$. However the last part $\{v,
u_{1,0},\ldots,u_{k-1,0} \}$ has cut-value $k-1$; so the resulting objective value is $k-1$.

On the other hand, it can be checked directly that the optimal value is at most {\em four}: Consider the partition
obtained by repeatedly taking the first $k$ consecutive vertices from the ordering $v,u_{1,0},\ldots,u_{1,k},
u_{2,0},\ldots,u_{2,k},\ldots,u_{k-1,0},\ldots,u_{k-1,k}$.

\end{document}